\documentclass[a4paper,12pt]{article}

\usepackage{graphicx}
\usepackage{listings}
\usepackage{color}
\usepackage{amssymb, amsmath, geometry}
\usepackage{natbib}

\usepackage{amsmath,amsthm}
\usepackage{graphicx}

\theoremstyle{definition}
\newtheorem{example}{Example}
\newtheorem{remark}{Remark}
\newtheorem{definition}{Definition}
\newtheorem{proposition}{Proposition}

\DeclareMathOperator*{\plim}{p-lim}

\begin{document}

\title{\bf Robust semiparametric inference with missing data}

\author{Eva Cantoni$^{1}$ and Xavier de Luna$^{2}$ \\
\\
$^{1}$ Research Center for Statistics and \\ Geneva School of Economics and 
Management, \\University of Geneva, Geneva 4, 1211, Switzerland.\\
$^{2}$ Department of Statistics, \\Ume{\aa} School of Business, Economics and Statistics, \\
 Ume{\aa} University, Ume{\aa}, SE-90187, Sweden.}
\providecommand{\keywords}[1]{\textbf{\textit{Keywords: }} #1}
 
\maketitle

\begin{abstract}
 Classical semiparametric inference with missing outcome data is not robust to contamination of the observed data and a single observation can have arbitrarily large influence on estimation of a parameter of interest. This sensitivity is exacerbated when inverse probability weighting methods are used, which may overweight contaminated observations. We introduce inverse probability weighted, double robust and outcome regression estimators of location and scale parameters, which are robust to contamination in the sense that their influence function is
bounded. We give asymptotic properties
and study finite sample behaviour.  Our simulated experiments show that contamination can be more serious a threat to the quality of inference than model misspecification. An interesting aspect of our results is that the auxiliary outcome model used to adjust for ignorable missingness by some of the estimators, is also useful to protect against contamination. We also illustrate through a case study how both adjustment to
ignorable missingness and protection against
contamination are achieved through weighting schemes, which can be contrasted to gain further insights.
\end{abstract}


\section{Introduction}

Many data analyses are concerned with drawing inference on a parameter $\beta$ partially characterising a distribution law of interest from which data is assumed to be a random sample. However, most often the observed data deviates from this ideal random sample scenario, for instance such that some of the observations are contaminated, i.e. drawn from a nuisance distribution law. Another common deviation is that the random sample is incomplete: some data are missing, due to dropout in follow up studies, non-response in surveys, etc. Such corrupted random samples are indeed the rule rather than the exception in applications, and we give a telling example in Section \ref{bmi.sec}, where we study BMI change in a ten year follow up study. While methods are available to deal with these two different problems separately as described below, it is essential to have inferential methods able to deal with situations where both types of corruption (missingness and contamination) arise simultaneously. Indeed, while it is well known that many estimating procedures, including OLS, ML, and method of moments, lack robustness to contamination \citep[a single observation can have arbitrarily large influence, e.g.,][]{Hamp:Ronc:Rous:Stah:1986,Heri:Cant:Copt:Vict:2009}, it is seldom acknowledged that this sensitivity to contamination can be exacerbated with estimators able to deal with missing data; see, however, \citet[][]{Hulliger:1995} and \citet{Beaumont:2013}, where this increased sensitivity has been pointed out in the context of surveys of finite populations. The potential increase sensitivity to contamination arises in particular for estimators overweighting some of the observations (those representing part of the data which are missing), and if the overweighted observations are by chance contaminated, this will have large negative impact on the inference. Thus, while robust methods are important in general, they are even more so when missing data needs to be accounted for.


In this paper we focus on situations where, while observations may be missing for the response, 
there is a set of background variables (covariates) which are observed for all units, and we can assume that outcomes are independent of missingness given the observed covariates (ignorable missingness assumption). Under the latter assumption, auxiliary (nuisance) models explaining the missingness mechanism and the outcome given the covariates can be combined in different ways to
obtain semiparametric estimators of $\beta$. Classical examples include inverse probability weighted estimators (IPW), using the missingness mechanism model as weights \citep{Hor:Thom:1952}, and augmented inverse probability weighted estimators \citep[AIPW,][]{robins1994estimation} using both auxiliary models. AIPW estimators are then robust to the misspecification of one of these two auxiliary models at a time (thus the name doubly robust estimator often used); see, e.g., \cite{tsiatis2006}, \cite{RotnitzkyStijn:15}. Finally, outcome regression imputation (OR) estimators using only the model for the outcome may also be used, thereby avoiding weighting \citep[][]{kang2007,Tan:2007}.

 
Within this context of ignorable missing data in the outcome, we introduce and study
estimators that are able to deal with situations where most of
 the units in the sample are randomly drawn from the distribution of interest
 while a smaller number of units is possibly drawn from another nuisance
 distribution. An estimator is considered robust to such contamination if it has
 bounded influence function, see \cite{Hamp:Ronc:Rous:Stah:1986}. This is because the influence function measures the asymptotic bias due to an infinitesimal contamination. A single observation can thus yield arbitrarily large bias if the influence function of the estimator is not bounded.  
 Classical IPW, AIPW and OR estimators have unbounded influence
 function. They are not robust in this sense, even though AIPW has a robustness property, but merely to misspecification of one of the auxiliary models used.  
 In a full data and finite parametric context, bounded influence function estimators are most  
 naturally introduced as M-estimators \citep{Huber:1964,Hamp:1974}. Here we take
 advantage of the fact that IPW, AIPW and OR estimators are partial M-estimators
 \citep{newey:mcfadden:94,Stef:Boos:2002,Zhel:Gent:Ronc:2012} to propose bounded influence function estimators. 
An interesting result of the introduced estimators is that the auxiliary outcome regression model used by AIPW to improve on efficiency compared to IPW, happens to also be useful in improving on the robustness properties of AIPW and OR. Robustness to contamination is typically obtained at the price of a loss in efficiency, although the latter can be controlled and set to say approximately 5\% under some conditions. On the other hand, our simulated experiments show that moderate contamination seriously affects the quality of classical semiparametric inference, more so than model misspecification. Our
 approach is general and we fully spell out the case where $\beta$ is the two dimensional
 location-scale parameter. 
  
 The paper is organized as follows. Section 2 presents formally the context, and introduces robust estimators for missing data situations, together with their asymptotic properties. Section 3 studies finite sample properties through simulation designs previously used by \cite{Lunceford:Davidian:2004}, to which we have added several contamination schemes. This allows us to study robustness due both to model misspecification and to contamination.       
 In Section 4 a longitudinal study of BMI based on electronic record linkage data is used to illustrate, e.g., how the robustification introduced can be seen as a weighting scheme which can be compared to the weighting used to correct for ignorable missingness. 
 The paper is concluded with a discussion in Section 5. Regularity conditions, proofs, implementation details and exhaustive results from the  simulations are relegated to the Appendix.

\section{Theory and method}\label{theory.sec}
\subsection{Notation and context}
Let a vector variable $Z$ be partitioned as $(Z_{2}',Z_1')'$, and consider the ideal situation when $(Z_{2i},Z_{1i}$, $i=1,\ldots,n)$ are independently drawn from a
probability law with density $p(Z_{2i},Z_{1i};\beta,\eta)=p(Z_{2i};\beta,\eta)p(Z_{1i}\mid Z_{2i};\eta)$ 
for unknown values $\beta=\beta_0$ and $\eta=\eta_0$, where $\beta$, of finite dimension, is the
parameter of interest describing some aspects of the distribution, $\eta$ is a nuisance parameter possibly of infinite
dimension, and $\beta$ and $\eta$ are variationally independent
\cite[semiparametric model; see][Chap.~4]{tsiatis2006}. We consider simultaneously two types of deviation from the above ideal random sample setting.

First, situations where atypical observations can occur in $Z_{2i}$ (and possibly $Z_{1i}$), i.e. where the majority of the data is generated as described above, but some of the observations may be issued from a different, but unknown, distribution. The final goal is to draw inference about $\beta$, even in the presence of a small fraction of spurious data
points. 

Further, we also want to allow for incomplete data situations, where we observe only
$(R_iZ_{2i},Z_{1i},R_{i}$, $i=1,\ldots,n)$, with $R_i$ a binary variable
indicating the observation status of $Z_{2i}$: $R_i=1$ if observed and $R_i=0$
if missing. 
We make throughout the missing at
random assumption (also called ignorable missingness), i.e. $\Pr(R_i=1\mid Z_{2i},Z_{1i})=\pi(Z_{1i})$, with $\pi(Z_{1i})>0$ on the support of $Z_{1i}$. 
The missing assignment mechanism is modelled up
to a parameter $\gamma$, $ \pi(Z_{1i};\gamma)$, and we distinguish cases where this model
is correctly specified, i.e. $\Pr(R_i=1\mid
Z_{1i})=\pi(Z_{1i};\gamma_0)$ for a given but unknown $\gamma_0$, and cases
where it is misspecified, i.e. an incorrect model for $\Pr(R_i=1\mid
Z_{1i})$ is used. 

\subsection{Full data case: robust M-estimators}
Let us first consider an estimating function $m(Z_2;\beta)$, which would be used if we had no missing data ($R_i=1$ for all $i$):
\begin{equation}\label{mestimator.eq}
\sum_{i=1}^n m(Z_{2i};\beta)=0.
\end{equation}
The choice of $m(Z_{2i};\beta)$ may be done based on desired properties for the
resulting M-estimator for $\beta$ (in the complete data case);
e.g., such that $E(m(Z_{2i};\beta_0))=0$ for consistency. The study of robustness properties to contamination 
was formalised in Hampel (1974). The influence function plays a central role because it can be interpreted as measuring the asymptotic bias due to an infinitesimal contamination. Here, the influence function for the resulting estimator $\hat\beta$ solution of (\ref{mestimator.eq}) is
\begin{align}
E\left( - \frac{\partial m(Z_{2i};\beta)}{\partial \beta} \right)^{-1} m(Z_{2i};\beta)
\end{align}
 under suitable regularity conditions \citep{Stef:Boos:2002}.

In the sequel we focus on the location-scale parameter 
$\beta=(\mu=E(Z_{2i}),\sigma^2=Var(Z_{2i}))'$. A commonly used choice is 
$m(Z_{2i};\beta)=(Z_{2i}-\mu,(Z_{2i}-\mu)^2-\sigma^2)'$, because the resulting estimator is efficient in the Gaussian case.
For this choice of $m$ estimating function, the influence function will not be bounded in $Z_{2i}$ and therefore not robust to contamination; see, e.g., \citet[Chap. 2]{Maro:Mart:Yoha:2006}. A general class of
M-estimators for $\mu$ and $\sigma^2$ are solution of (\ref{mestimator.eq}) for 
\begin{align}\label{psi.eq}
m_\psi(Z_{2i};\beta)=& 
\left(\begin{array}{c}
 \psi_{c_{\mu}} \left( \frac{Z_{2i}- \mu}{\sigma}\right)-A  \\
 \psi_{c_{\sigma}}^2\left( \frac{Z_{2i}- \mu}{\sigma}
\right)- B 
\end{array}\right),
\end{align}
where $\psi_{c}(\cdot)$ is an odd function, and where
$A= E\left\{\psi_{c_{\mu}}\left( \sigma_0^{-1} (Z_{2i}- \mu_0) \right)\right\}$ 
and 
$B= E\left\{\psi_{c_{\sigma}}^2\left(\sigma_0^{-1} (Z_{2i}- \mu_0)\right)\right\}$ in order
to ensure that $E(m_\psi(Z_{2i};\beta))=0$ at $\beta_0=(\mu_0,\sigma_0)$, the
true unknown value for $(\mu,\sigma)$.  
Bounded influence function estimators are obtained by using bounded $\psi_{c}(\cdot)$ functions, e.g., the  
Huber function $\psi_c(t) = \min\{ c,\max\{ t,-c\}\} $, and the Tukey biweight
function  
$\psi_c(t) = ((t/c)^2-1)^2 t $ if $|t|<c$ and $0$ otherwise, see
\cite{Heri:Cant:Copt:Vict:2009} for further details. The value for $c$ can be chosen appropriately
to control efficiency under the non-contaminated Gaussian
case.
Equations (\ref{mestimator.eq}) using (\ref{psi.eq}) need to be solved simultaneously for $\mu$ and $\sigma$. 

\subsection{Robust estimation with missing data}
Semiparametric estimation with missing data has been reviewed for instance in \citet{tsiatis2006}. We introduce below novel bounded influence function estimators. 

Let $ \pi(Z_{1i};\gamma)$ be a well specified parametric model, i.e. such that for 
$\Pr(R_i=1\mid Z_{1i})=\pi(Z_{1i};\gamma_0)$ for an unknown value $\gamma_0$.
Assume that we have an estimator 
$\hat\gamma$ of $\gamma$ solution of estimating equations
\begin{align}
\label{estim.gamma}
&\sum_{i=1}^n m_\gamma (R_i,Z_{1i};\gamma)=0,
\end{align}
 such that $\plim_{n\rightarrow\infty}\hat\gamma = \gamma_0$. 

\begin{definition}
A robust inverse probability weighted (RIPW) estimator \\
$(\hat\mu_{RIPW},\hat\sigma_{RIPW})'$ of $(\mu,\sigma)'$ is solution of the estimating
equation: 
\begin{align}\label{robustmu.ipw.eq}
\sum_{i=1}^n \varphi_{RIPW}(Z_{i},R_i;\beta,\hat\gamma)=0,
\end{align}
where
\begin{align}\nonumber
\varphi_{RIPW}(Z_{i},R_i;\beta,\gamma)=\left(
\begin{array}{c}
\frac{R_i \big( \psi_{c_{\mu}}\big( \sigma^{-1}
(Z_{2i}- \mu)\big)-A\big)}{\pi(Z_{1i};\gamma)} \\
 \frac{R_i \big( \psi_{c_{\sigma}}^2\big( \sigma^{-1}
(Z_{2i}- \mu)\big)-B\big)}{\pi(Z_{1i};\gamma)} 
\end{array}\right),
\end{align}
with  $A=E\left\{ \psi_{c_{\mu}}\big( \sigma_0^{-1}
(Z_{2i}- \mu_0)\big)\right\}$ and $
B = E\left\{  \psi_{c_{\sigma}}^2\big( \sigma_0^{-1}
(Z_{2i}- \mu_0)\big)\right\}.
$
\end{definition}

A similar estimator was proposed and studied in \cite{Hulliger:1995} in the
context, however, of finite populations and surveys. Note that letting $\psi_c(t)
= t$, the identity function, yields a classical inverse probability weighted
estimator \citep{Hor:Thom:1952}.

\begin{remark}\label{weight.rem}
	 RIPW estimation can be interpreted as a double weighting scheme estimator, where $Z_{2i}$ observations are weighted with inverse propensity scores $1/\pi(Z_{1i};\gamma)$ (i.e., observations lying on the covariate support where the probability of dropout is higher are overweighted) and with $\psi$ weights $\psi_{c_{\mu}}\big( \sigma^{-1} (Z_{2i}- \mu)\big) / (\sigma^{-1}
(Z_{2i}- \mu))$ (i.e., outlying observations are downweighted).  
These weights as well as the compound weights $1/\pi(Z_{1i};\gamma)  \times \psi_{c_{\mu}}\big( \sigma^{-1}
(Z_{2i}- \mu)\big) / (\sigma^{-1}
(Z_{2i}- \mu))$ may be looked at in applications to gain insight in how the two weighting schemes interact; see Section \ref{bmi.sec} for an illustration.
\end{remark}

\begin{proposition}\label{ripw.prop}
Let $\pi(Z_{1i};\gamma)$ be correctly specified with (\ref{estim.gamma}) such
that
$\plim_{n\rightarrow\infty}\hat\gamma= \gamma_0$. Then, under regularity
conditions given in Appendix~\ref{prop1and3.sec}, $(\hat\mu_{RIPW},\hat\sigma_{RIPW})'$ is consistent for $(\mu_0,\sigma_0)'$ and has the following asymptotic multivariate normal distribution as $n\rightarrow\infty$
$$\sqrt n \Big( (\hat\mu_{RIPW},\hat\sigma_{RIPW})'-(\mu_0,\sigma_0)'\Big )\overset{d}\rightarrow N\big(0,E\big\{IF_{RIPW}(IF_{RIPW})'\big\}\big),$$
where  $IF_{RIPW}$
is the influence function:
\begin{align}
\lefteqn{IF_{RIPW}(Z_{i},R_i; \beta)= -\left\{ E \left[ \frac{\partial m_\psi(Z_{2i};\beta)}{\partial \beta'} \right] \right\}^{-1} \Bigg\{ \varphi_{RIPW}(Z_{i},R_i;\beta,\gamma_0)} \nonumber \\
&- E\left[ \frac{\partial \varphi_{RIPW}(Z_{i},R_i; \beta,\gamma_0) }{\partial \gamma'}\right]\left\{E\left[ \frac{\partial m_\gamma (R_i,Z_{1i};\gamma_0)}{\partial\gamma'}\right]\right\}^{-1}m_\gamma (R_i,Z_{1i};\gamma_0) \Bigg\}. \label{ifipw.eq}
\end{align}
\end{proposition}
Thus, from (\ref{ifipw.eq}) we see that the influence function of RIPW is
bounded in $Z_{2i}$ if the function $\psi_c(\cdot)$ is bounded. This is not the
case for the classical IPW, cor responding to $\psi_c(t) = t$.

The implementation of RIPW requires the
computation of $A$ and $B$. If the standardized quantity $\sigma_0^{-1} (Z_{2i}-
\mu_0)$ is satisfactorily approximated by a ${\mathcal N}(0,1)$ variate, then
$A=0$ (since $\psi_c$ is odd) and $B$ can be approximated by Monte Carlo
simulations.

In an attempt to improve efficiency one may consider $h(Z_{1i};\beta,\xi)$ a
working model (parametrised with $\xi$) for $E(m(Z_{2i};\beta)\mid
Z_{1i})$. This model is correctly specified for $E(m(Z_{2i};\beta)\mid Z_{1i})$
if $h(Z_{1i};\beta,\xi_0)=E(m(Z_{2i};\beta)\mid Z_{1i})$ for a 
value $\xi_0$. However, we call it working model because we will also consider situations where it is not necessarily correctly specified.
 Assume we have estimators $\hat\xi$ of $\xi$ and $\hat\gamma$ of $\gamma$, respectively solutions of (\ref{estim.gamma}) and 
\begin{align}
\label{estim.xi}
&\sum_{i=1}^n R_i m_\xi (Z_{i};\xi)=0,
\end{align}
such that
$\plim_{n\rightarrow\infty}\hat\xi = \xi^*$ and $\plim_{n\rightarrow\infty}\hat\gamma =
\gamma^*$, for some fix values $\xi^*$ and $\gamma^*$. In the correctly specified cases $\xi^*=\xi_0$ and $\gamma^*=\gamma_0$.
\begin{definition}	
A robust augmented IPW (RAIPW) estimator
$(\hat\mu_{RAIPW},\hat\sigma_{RAIPW})'$ of $(\mu,\sigma)'$ is solution of the estimating equation: 
\begin{align}
\sum_{i=1}^n \varphi_{RAIPW}(Z_{i},R_i;\beta,\hat\gamma,\hat\xi) = 0, 
\label{robusts.aipw.eq} 
\end{align}
where
\begin{align}\nonumber
\varphi_{RAIPW}(Z_{i},R_i;\beta,\gamma,\xi)=\left(
\begin{array}{c}
\frac{R_i \big( \psi_{c_{\mu}}\big( \sigma^{-1}
(Z_{2i}- \mu)\big)-A\big)}{\pi(Z_{1i};\gamma)}-\left[
\frac{R_i-\pi(Z_{1i};\gamma)}{\pi(Z_{1i};\gamma)} h_1(Z_{1i};\beta,\xi) \right] \\
 \frac{R_i \big( \psi_{c_{\sigma}}^2\big( \sigma^{-1}
(Z_{2i}- \mu)\big)-B\big)}{\pi(Z_{1i};\gamma)}- \left[
\frac{R_i-\pi(Z_{1i};\gamma)}{\pi(Z_{1i};\gamma)} h_2(Z_{1i};\beta,\xi) \right] 
\end{array}\right),
\end{align}
$h_1(Z_{1i};\beta,\xi)$ is a working model for 
$E\Big(  \psi_{c_{\mu}}(
\sigma^{-1} 
(Z_{2i}- \mu)\big) -A | Z_{1i} \Big)$ and $h_2(Z_{1i};\beta,\xi)$ for 
$E\Big(  \psi_{c_{\sigma}}^2(
\sigma^{-1} 
(Z_{2i}- \mu)\big) -B | Z_{1i} \Big)$, and  $A= E\left\{ \psi_{c_{\mu}}\big( \sigma_0^{-1}
(Z_{2i}- \mu_0)\big)\right\}$ and $
B = E\left\{  \psi_{c_{\sigma}}^2\big( \sigma_0^{-1}
(Z_{2i}- \mu_0)\big)\right\}.
$
\end{definition}

Using the identity function for $\psi_c$ yields a classical augmented inverse
probability weighting (AIPW) estimator  \citep{robins1994estimation}. 

\begin{proposition}\label{raipw.prop}
Let $\pi(Z_{1i};\gamma)$ be correctly specified with (\ref{estim.gamma}) such
that \\
$\plim_{n\rightarrow\infty}\hat\gamma=\gamma_0$ and/or let
$h(Z_{1i};\beta,\xi)=(h_1(Z_{1i};\beta,\xi),h_2(Z_{1i};\beta,\xi))'$ be
correctly specified with (\ref{estim.xi}) such that
$\plim_{n\rightarrow\infty}\hat\xi=\xi_0$. Then, under regularity conditions
given in Appendix~\ref{proof.section},	 
 $(\hat\mu_{RAIPW},\hat\sigma_{RAIPW})$ is consistent for $(\mu_0,\sigma_0)'$ and has the following asymptotic multivariate normal distribution as $n\rightarrow\infty$
$$\sqrt n \Big( (\hat\mu_{RAIPW},\hat\sigma_{RAIPW})'-(\mu_0,\sigma_0)'\Big )\overset{d}\rightarrow N\big(0,E\big\{IF_{RAIPW}(IF_{RAIPW})'\big\}\big),$$
where 
\begin{align}
\lefteqn{IF_{RAIPW}(Z_{i},R_i; \beta)= -\left\{ E \left[ \frac{\partial m_\psi(Z_{2i},\beta)}{\partial \beta'} \right] \right\}^{-1}\Bigg\{ \varphi_{RAIPW}(Z_{i},R_i; \beta,\gamma^*,\xi^*)} \nonumber \\
&- E\left[ \frac{\partial \varphi_{RAIPW}(Z_{i},R_i; \beta,\gamma^*,\xi^*)}{\partial \gamma'} \right]\left\{E\left[ \frac{\partial m_\gamma (R_i,Z_{1i};\gamma^*)}{\partial\gamma'}\right]\right\}^{-1}m_\gamma (R_i,Z_{1i};\gamma^*) \nonumber \\
&- E\left[ \frac{\partial \varphi_{RAIPW}(Z_{i},R_i;\beta,\gamma^*,\xi^*) }{\partial \xi'}\right]\left\{E\left[ \frac{\partial m_\xi (Z_{i};\xi^*)}{\partial\xi'}\right]\right\}^{-1}m_\xi (Z_{i},\xi^*)\Bigg \}. \label{ifaipw.eq}
\end{align}

\end{proposition}
Thus,  RAIPW is as AIPW doubly robust in the sense that only one of the two
auxiliary models used must be correctly specified in order to obtain consistent
and asymptotic normal estimators. Moreover, the influence function of RAIPW is
bounded in $Z_{2i}$ if the function $\psi_c(\cdot)$ is bounded (assuming the
estimating
equation \eqref{estim.xi} of the auxiliary model has also bounded influence
function in $Z_{2i}$; see Exemple \ref{example.sls}),  while this is not the
case for the classical AIPW.

\begin{example}[RAIPW estimator of location and scale]\label{example.sls}
Let us specify a working model parametrized by $\xi=(\xi_1',\xi_2)'$ as 
\begin{align}\label{auxi.mod}
	Z_{2i}=\tilde h(Z_{1i};\xi_1)+\xi_2\nu,
\end{align}
with $\nu\sim N(0,1$).
Note that this does not constrain $Z_{2i}$ to have a symmetric distribution as was the case for RIPW. 
The corresponding working model for $E(m(Z_{2i};\beta)\mid Z_{1i})$ is such that $h_1(Z_{1i};\beta,\xi)=\tilde h(Z_{1i};\xi_1)-\mu$ and $h_2(Z_{1i};\beta,\xi)=(\tilde h(Z_{1i};\xi_1)-\mu)^2+\xi_2^2-\sigma^2$.
Estimators of $\xi$ with bounded influence function in this context are described in Appendix~\ref{robreg.section}.
Then, 
\begin{align}
	h_1(Z_{1i};\beta,\xi)&=E\Big( \psi_{c_{\mu}}(
\sigma^{-1} 
(\tilde h(Z_{1i};\xi_1)+\xi_2\nu- \mu)\big) | Z_{1i} \Big)-A, \label{workmu.eq}  \\
h_2(Z_{1i};\beta,\xi)&=E\Big(  \psi_{c_{\sigma}}^2(
\sigma^{-1} 
(\tilde h(Z_{1i};\xi_1)+\xi_2\nu- \mu)\big) | Z_{1i} \Big)-B, \label{works.eq}
\end{align}
may be computed using numerical integration for the conditional expectations, where 
$E(\cdot\mid Z_{1i})$ is the expectation under model~\eqref{auxi.mod}. Under the latter model, 
$A=0$ and Monte Carlo simulations can be used to obtain B. 
Both (\ref{workmu.eq}) and (\ref{works.eq}) can be used to obtain RAIPW estimators through \eqref{robusts.aipw.eq}. See Appendix~\ref{implementation.section} for more implementation details.   
\end{example}

Finally, when the outcome model is correctly specified yet another robust estimator can be introduced.
\begin{definition}
A robust outcome regression estimator (ROR) $(\hat\mu_{ROR},\hat\sigma_{ROR})'$ of ($\mu,\sigma$) is solution of 
 \begin{align}\label{imp.eq}
&\sum_{i=1}^n  h(Z_{1i};\beta,\hat\xi)  =\sum_{i=1}^n  \varphi_{ROR}(Z_{1i};\beta,\hat\xi)=0
\end{align}
by using a correctly specified working model $h(Z_{1i};\beta,\xi_0)=E(m(Z_{2i};\beta)\mid Z_{1i})$ together with $\hat\xi$,
an M-estimator (\ref{estim.xi}) of $\xi$ with bounded influence function.
\end{definition}

\begin{proposition}\label{ror.prop}
Let $h(Z_{1i};\beta,\xi)=(h_1(Z_{1i};\beta,\xi),h_2(Z_{1i};\beta,\xi))'$ be correctly specified with (\ref{estim.xi}) such that $\plim_{n\rightarrow\infty}\hat\xi=\xi_0$. Then, under regularity conditions given in Appendix~\ref{prop1and3.sec},	
 $(\hat\mu_{ROR},\hat\sigma_{ROR})'$ is consistent for $(\mu_0,\sigma_0)'$ and has the following asymptotic multivariate normal distribution as $n\rightarrow\infty$
$$\sqrt n \Big( (\hat\mu_{ROR},\hat\sigma_{ROR})'-(\mu_0,\sigma_0)'\Big )\overset{d}\rightarrow N(0,IF_{ROR}(IF_{ROR})'),$$
where 
\begin{align}
IF_{ROR}(Z_{i},R_i; \beta)=&-\left\{ E \left[ \frac{\partial \varphi_{ROR}(Z_{1i}; \beta,\xi_0)}{\partial \beta'} \right] \right\}^{-1}\Bigg\{ \varphi_{ROR}(Z_{1i}; \beta,\xi_0) \nonumber \\
&- E\left[ \frac{\partial \varphi_{ROR}(Z_{1i};\beta,\xi_0) }{\partial \xi'}\right]\left\{E\left[ \frac{\partial m_\xi (Z_{i};\xi_0)}{\partial\xi'}\right]\right\}^{-1}m_\xi (Z_{i},\xi_0)\Bigg \}. \label{ifor.eq}
\end{align}
\end{proposition}

\begin{example}[ROR estimator of location and scale]\label{example.ror}
Within the context of Example~\ref{example.sls}, assume that model (\ref{auxi.mod}) holds. Then, 
$h(Z_{1i};\beta,\xi)=(\tilde h(Z_{1i};\beta,\xi)-\mu,(\tilde h(Z_{1i};\xi_1)-\mu)^2+\xi_2^2-\sigma^2)'$, and $\xi$ is estimated with a bounded influence function estimator; see
Appendix~\ref{robreg.section} for details.
\end{example}

Unlike for RAIPW, the regularity conditions apply to the working
model $h(Z_{1i};\beta,\xi)$ only. For
instance, to characterise the influence function 
one need to be more specific about the working model (which needs to be
correctly specified).  
On the other hand, the results of Propositions \ref{ripw.prop} and \ref{raipw.prop} for RIPW and RAIPW respectively give specifically the regularity conditions that must apply to the $\psi_c$ functions used, and the influence functions resulting.

We have focused on robustness properties to contamination in the outcome $Z_{2i}$. Contamination in the covariates $Z_{1i}$ may also happen. This is typically tackled by using the Tukey's redescending $\psi$ function which protect against high leverage points, i.e. outlying values in the design space; see, e.g., \citet[Chap. 4 and 5]{Maro:Mart:Yoha:2006} and \citet{Cant:Ronc:2001}.

\section{Simulation experiments}
\label{sim.section}
We present a large simulation exercise to assess several aspects of our
procedure for the joint estimation of location and scale:
behaviour for clean data, robustness to the presence of contamination, and sensitivity to
model misspecification.

\subsection{Simulation setting}
\label{simsetting.section}

We implement the same simulation design as \cite{Lunceford:Davidian:2004}. We consider the covariates $X = (X_1,X_2,X_3)'$
associated with both the missingness mechanism and the outcome, and
the covariates $V= (V_1,V_2,V_3)'$ which are associated only with the outcome. The variables $(X_{1},X_{2}, X_{3}, V_{1}, V_{2}, V_{3})^\prime$ are realizations of the joint distribution of $(X',V')'$ built by first taking $X_3 \sim \mbox{Bernoulli}(0.2)$. Then,
conditionally on $X_3$, $V_3$ is generated as Bernoulli with $\Pr(V_3=1 \mid X_3)
= 0.75 X_3 + 0.25 (1-X_3)$ and finally $(X_1, V_1, X_2 ,V_2)' \mid X_3$ is taken
from a multivariate normal distribution ${\mathcal N}(\tau_{X_3},\Sigma_{X_3})$,
where  
$\tau_1 = (1,1,-1,-1)'$, $\tau_0 = (-1,-1,1,1)'$ and
$$\Sigma_1 = \Sigma_0 = \left( \begin{array}{cccc}
1 & 0.5 & -0.5 & -0.5 \\
0.5 & 1 & -0.5 & -0.5\\
-0.5 & -0.5 & 1 & 0.5 \\
-0.5 & -0.5 & 0.5 & 1 
\end{array} \right).$$ 

For each individual $i=1,\ldots,n$, the missingess mechanism indicator $R_i$ is generated as a Bernoulli
variable with probability of missingness ($R_i=0$) defined by 
\begin{equation*}
\Pr(R_i =0\mid X,V)  = \frac{\exp( \gamma_{1} +\gamma_{2} X_{1i} + \gamma_{3} X_{2i} +\gamma_{4} X_{3i}) }{1+\exp( \gamma_{1} +\gamma_{2} X_{1i} + \gamma_{3} X_{2i} +\gamma_{4} X_{3i})},
\end{equation*}
which corresponds to the control group in \cite{Lunceford:Davidian:2004}.

The response $Z_{2i}$ is generated according to the model
\begin{equation}
\label{sim.model}
Z_{2i} = \xi_{10} +\xi_{11} X_{1i} + \xi_{12} X_{2i} +\xi_{13} X_{3i} + \xi_{14} V_{1i} + \xi_{15} V_{2i} + \xi_{16} V_{3i} + \epsilon_i,
\end{equation}
where $\epsilon_i \sim {\mathcal N}(0,\xi_2^2 =1)$ and in our notation $\xi_1=(\xi_{10},\xi_{11},\cdots,\xi_{16})$.

The parameter values $(\xi_{10},\xi_{11}, \xi_{12},\xi_{13} )' =
(0,-1,1,-1)'$ are kept fixed throughout, whereas different scenarios are
considered for $(\xi_{14},\xi_{15},\xi_{16})'$ and $\gamma$, namely
\begin{equation}
\label{csi.value}
(\xi_{14},\xi_{15},\xi_{16})'= \left\{ 
\begin{array}{ll}
(-1,1,1)' & \mbox{strong association}\\
(-0.5,0.5,0.5)' & \mbox{moderate association}\\
(0,0,0)' & \mbox{no association}
\end{array} \right.
\end{equation}
and
\begin{equation}
\label{gamma.value}
 \gamma = (\gamma_1,\gamma_2,\gamma_3,\gamma_4)' = \left\{ 
\begin{array}{ll}
(0,0.6,-0.6,0.6)' & \mbox{strong association}\\
(0,0.3,-0.3,0.3)' & \mbox{moderate association.}
\end{array} \right.
\end{equation}

Notice that when $(\xi_{14},\xi_{15},\xi_{16})' =(0,0,0)'$, $V$ is associated with neither the outcome nor the missingness mechanism. 
The values of $\xi$ and
$\gamma$ are such that lower response values and lower 
probabilities of missingness are obtained when $X_3=1$, and conversely when
$X_3=0$. 

We generate $1000$ realisations of size $n=1000$ and $5000$, called clean datasets, i.e. free of contamination. We present results for $n=1000$, while the larger sample size confirmed the results and are omitted. Departing from the clean datasets, we obtain corresponding contaminated datasets according to different schemes as we describe in Section~\ref{simcontam.section}.

The combination of parameters in \eqref{csi.value} and \eqref{gamma.value} gives
six designs.  
For each design, we fit a total of 20 estimators of $\beta = (\mu,
\sigma)'$. They differ in the choice of estimation strategy (IPW, AIPW, OR), whether they are in their classical or robust versions, and whether the auxiliary models are misspecified or not. Thus, we consider 
\\ \vskip-3mm
\indent IPW($X$), AIPW($X,X$), AIPW($X,XV$), OR($X$) and OR($XV$), 
\\ \vskip-3mm
\noindent and their robust versions 
\\ \vskip-3mm
\indent RIPW($X$), RAIPW($X,X$), RAIPW($X,XV$), ROR($X$) and ROR($XV$),
\\ \vskip-3mm
\noindent where the covariate sets used in the auxiliary models are given within parentheses, and, e.g., AIPW($X,XV$), means that the first set $X$ is used to explain $R_i$ and the second set $XV:= (X,V)$ is used to explain $Z_{2i}$. All these estimators use well specified auxiliary models. We, moreover, consider estimators using misspecified auxiliary models as follows: 
\\ \vskip-3mm
\indent IPW($X_{\_}$), AIPW($X_{\_},XV$), AIPW($X,X_{\_}V$), AIPW($X_{\_},X_{\_}V$), and OR($X_{\_}V$),
\\ \vskip-3mm
\noindent and their robust versions 
\\ \vskip-3mm
\indent RIPW($X_{\_}$), RAIPW($X_{\_},XV$), RAIPW($X,X_{\_}V$), RAIPW($X_{\_},X_{\_}V$), \\ \indent and ROR($X_{\_}V$), 
\\ \vskip-3mm
\noindent where
$X_{\_}:=X\setminus X_1$ and $X_{\_}V:=(X_{\_},V)$.
Auxiliary models explaining $R_i$ and $Z_{2i}$  are fitted using, respectively, logistic regression and ordinary least squares for the classical versions, and robust logistic regression and robust linear regression for the robust versions. For RIPW and RAIPW estimators Tukey's $\psi$ function is used in \eqref{robusts.aipw.eq}. Tukey's $\psi$ function is usually preferred over Huber's with asymmetric contamination.
The robust estimators are tuned to have approximately $95\%$ efficiency at the correctly specified models for clean data. The values of the corresponding tuning constants are given in Appendix~\ref{tuning.section}.
For details on the computation see Appendix~\ref{implementation.section}.

\subsection{Results for clean data}

The top half of Figure~\ref{CleanContam2XiModerateGammaModerateFig} summarises with boxplots the estimates of $\mu$  (left)
and $\sigma$ (right) for clean data, i.e.\  when the 1000 replicates are generated from the design introduced in Section~\ref{simsetting.section}, with $\gamma$ moderate and $\xi$ moderate.
The first row of panels show that for both $\mu$ and $\sigma$ all the estimators (classical and robust) except RIPW are, as expected, unbiased. The bias of RIPW is due to the correction terms ($A$ and $B$ in \eqref{robustmu.ipw.eq}) which are in this setting badly approximated based on the assumption that $Z_{2i}$ is normally distributed. This is improved for RAIPW, because the use of the outcome model allows for a better approximation of the correction terms. Also as expected, (R)IPW is more variable than (R)AIPW. 

\begin{figure}[h!]
\centering
\makebox{\includegraphics[width=\textwidth]{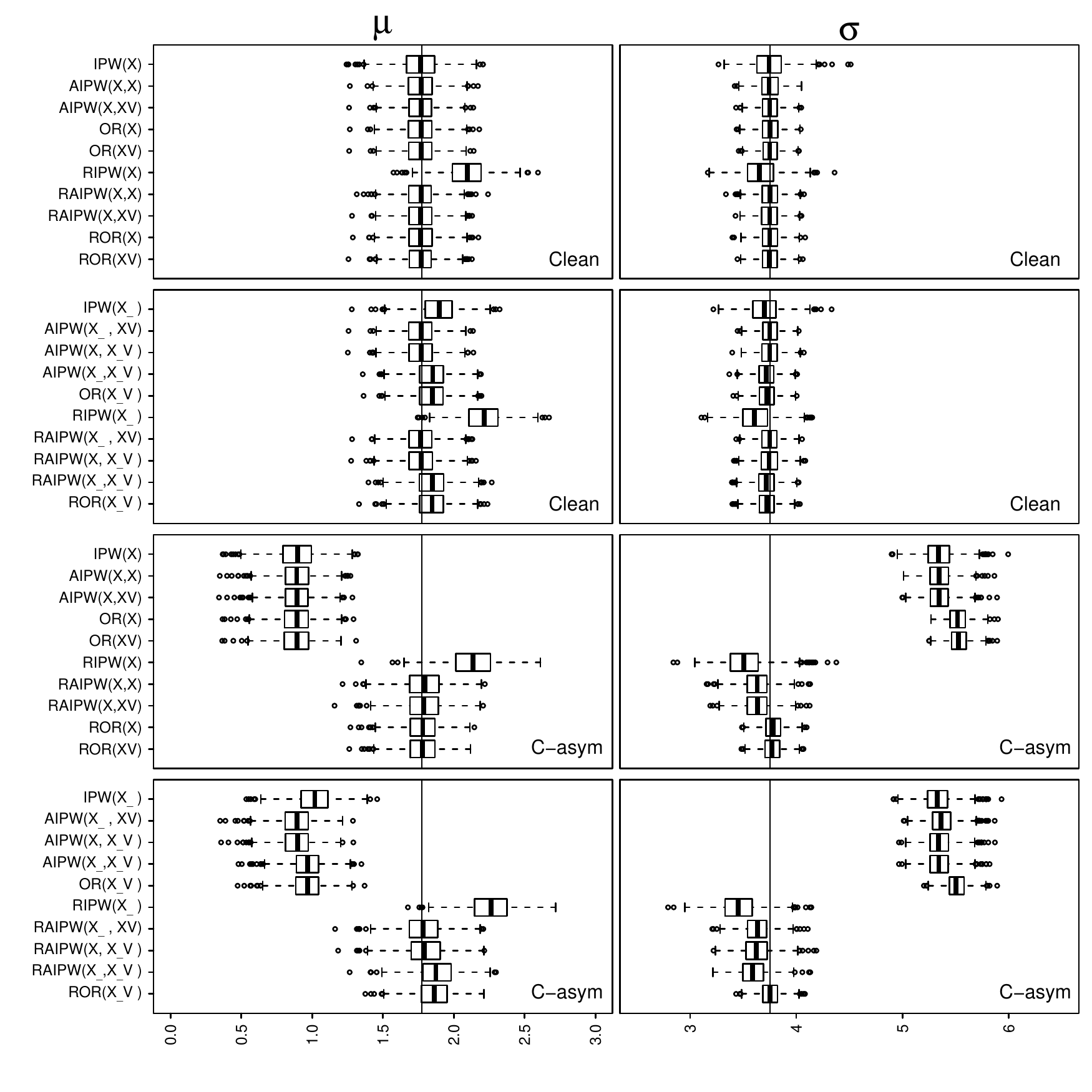}}
\caption{\label{CleanContam2XiModerateGammaModerateFig}Estimates of $\mu$  (left) and $\sigma$ (right) for the $\gamma$  moderate-$\xi$ moderate scenario for clean data and under the C-asym contamination.  The vertical lines represent the true underlying values.}
\end{figure}

The second row of panels confirms some other well known properties of the (A)IPW estimators: the bias due to misspecification of the missingness mechanism for IPW, the double robustness property of AIPW (i.e.,\ unbiasedness if only one of the auxiliary models is misspecified) and the sensitivity of the OR estimator to the misspecification of the outcome regression model. Essentially, these properties are preserved for the robust versions introduced herein. These results are also summarised numerically in Table~\ref{CleanModerateModerateTable} (Appendix~\ref{simulationmoderatemoderate.section}), where bias, standard deviations and root mean squared error of the estimators are reported. 
The results for the other five combinations of parameters (\eqref{csi.value} - \eqref{gamma.value}) deliver a similar general message,  with different magnitudes. The corresponding figures and tables supporting this claim are provided in Appendix~\ref{simulationothers.section}. 

\subsection{Results under contamination}
\label{simcontam.section}
With the result above that expected behaviours are obtained with clean data, we study now the effect on estimation of deviations from the data generating mechanism of interest.
To generate a contaminated sample, $5\%$ of the observed responses (i.e.\  data points for which $R_i=1$) issued from model~\eqref{sim.model} were randomly chosen and changed to the realization of:
\begin{description}
\item[C-asym] $U \sim {\mathcal U}(-20,-12)$, 
\item[C-sym] $W=BU-(1-B)U$, where $B$ is Bernoulli(probability=0.5),
\item[C-hidden] $N \sim {\mathcal N}(-10,0.4)$.
\end{description}
For the C-asym case, the range of the uniform distribution has been set such that it falls approximately outside the observed range of clean $Z_{2i}$. C-sym is the symmetric version of C-asym, and the C-hidden case is such that the contamination is not clearly visible when looking at the observed $Z_{2i}$ marginally. Figure~\ref{ContamSchemesFig} displays a realization of each scheme for the scenario with $\gamma$ moderate and $\xi$ moderate: the values of $Z_{2i}$ are plotted against $E(Z_{2i} \mid Z_{1i})$, the linear predictor of the outcome model~\eqref{sim.model}, with a histogram of the marginal distribution of $Z_{2i}$.

\begin{center}
\begin{figure}

\vskip 0.3cm
\hspace*{-0.5cm}
\begin{tabular}{ccc}
C-asym & C-sym & C-hidden  \\
\includegraphics[width=0.33\textwidth]{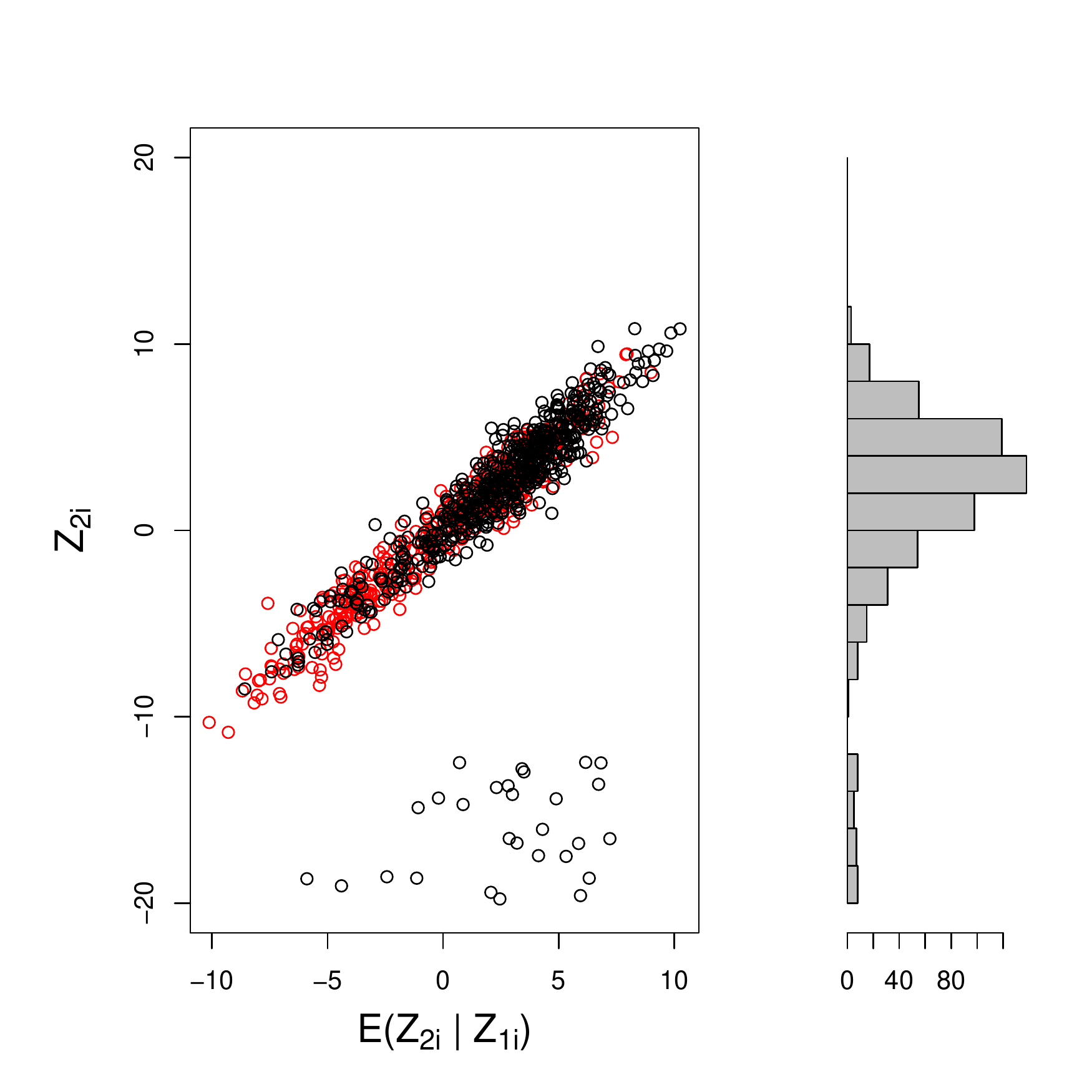} & \includegraphics[width=0.33\textwidth]{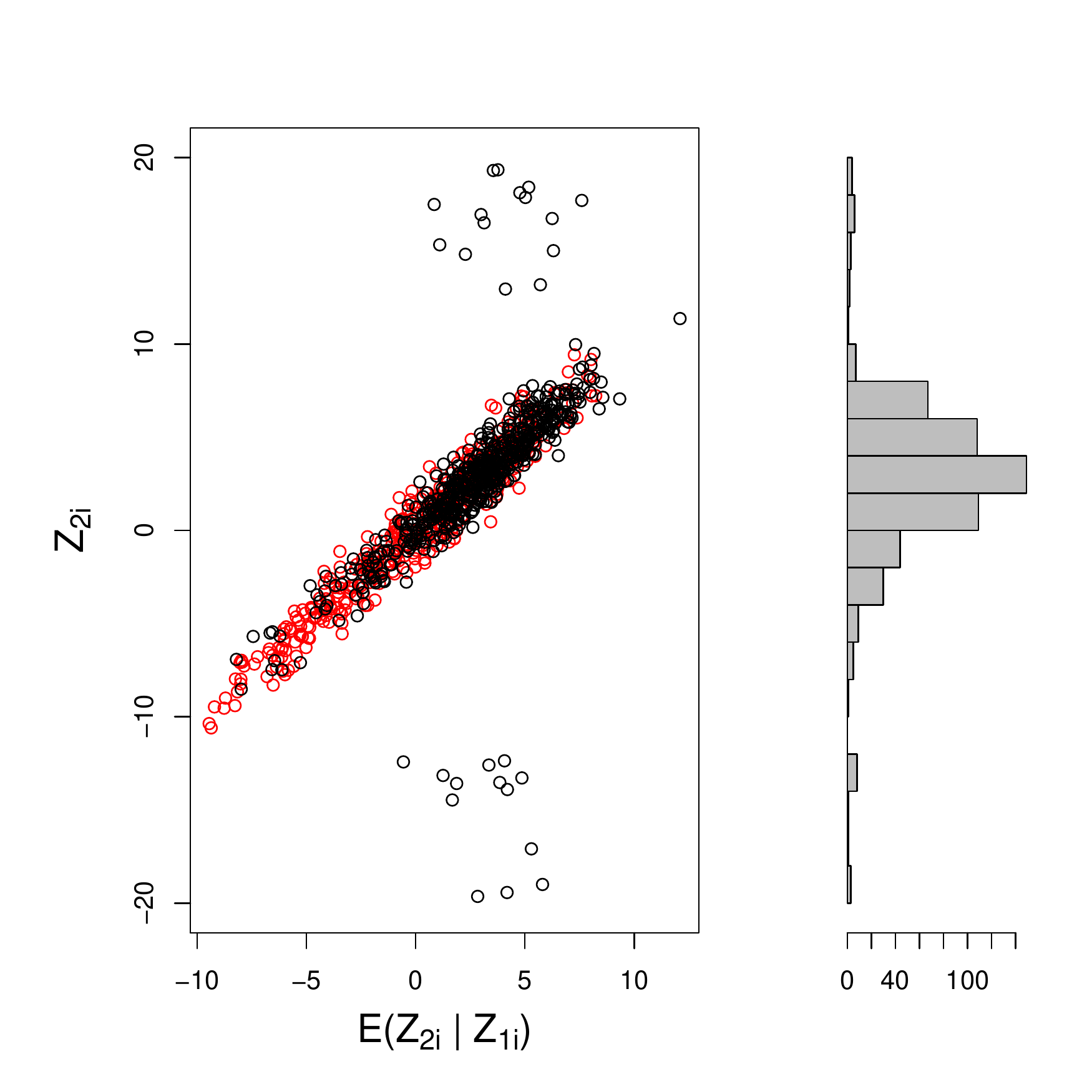} &
\includegraphics[width=0.33\textwidth]{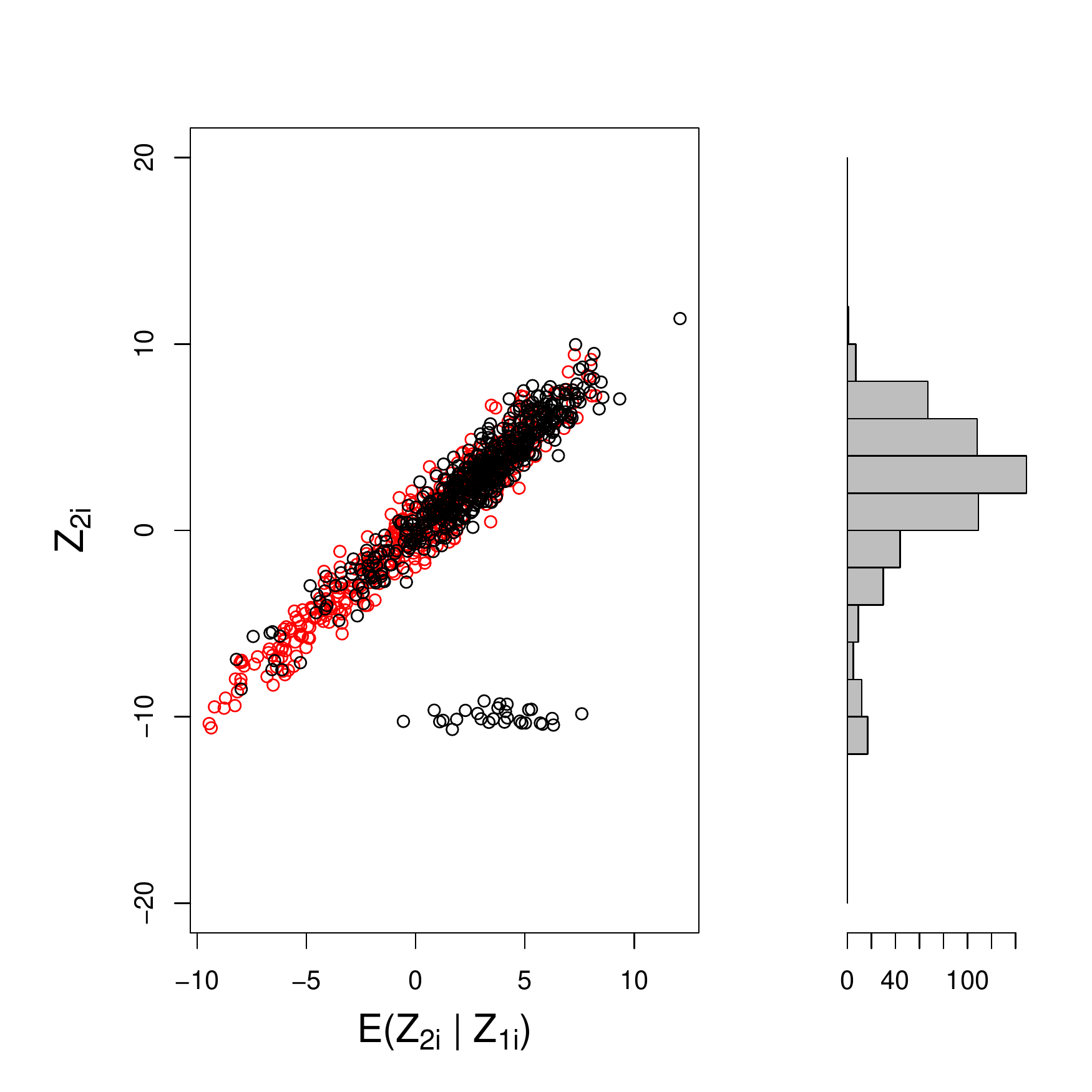}\\
\end{tabular}
\caption{One realization of size 1000 for the three contamination schemes considered: black circles observed outcomes ($R_i=1$) and circles unobserved outcomes ($R_i=0$). The histograms are over the observed outcomes $Z_{2i}$.} 
\label{ContamSchemesFig}
\end{figure}
\end{center}

We present the results for the $\gamma$ moderate and $\xi$ moderate design and contamination C-asym in bottom half of Figure~\ref{CleanContam2XiModerateGammaModerateFig}. The results for the other $\gamma$-$\xi$ combinations are given in Appendix~\ref{simulationothers.section}. The third row of panels of Figure~\ref{CleanContam2XiModerateGammaModerateFig} displays the results for the correctly specified models. We can see that for $\mu$ all the classical methods suffer a negative bias (underestimation) due to the presence of contamination, and these bias are of similar magnitude. For $\sigma$, the biases of the classical methods are positive (overestimation), with OR estimators being even more affected than (A)IPW estimators. On the other hand, RAIPW and ROR perform well, producing estimates in target with the true underlying values. A slight negative bias remains for RAIPW(X,XV) for $\sigma$. However, the size of this bias is negligible compared to the bias induced by contamination on the classical estimators. 

The fourth row of panels of  Figure~\ref{CleanContam2XiModerateGammaModerateFig}  shows the effects of both misspecification and contamination. We observe that the bias due to contamination is more severe than bias due to misspecification in the setting simulated (compare with the second row of the same figure, i.e. clean data). 

Root mean squared errors (RMSE) are given in Table \ref{Contam2ModerateModerateTable}, yielding further insights. In particular, while AIPW and OR with correct model specification were comparable in terms of empirical RMSE in the clean data designs (see RMSE tables in Appendix~\ref{simulationmoderatemoderate.section}), we observe that ROR outperforms RAIPW in this contamination case, particularly so when estimating $\sigma$ (upper half of Table \ref{Contam2ModerateModerateTable}). In fact we see that ROR has both lower bias and variance in this setting. When model misspecification occurs (lower half of Table \ref{Contam2ModerateModerateTable}), ROR also outperforms RAIPW if the latter misspecifies both auxiliary models, otherwise RAIPW has lowest RMSE for estimation of $\mu$ when one of the model used is correct. 

Results for the other parameter combinations under C-asym contamination carry similar messages, with different magnitudes; see figures and tables in Appendix~\ref{simulationothers.section}.

\begin{table}

\caption{\label{Contam2ModerateModerateTable} Bias, standard deviation and root mean squared error (times 10) of the estimates across simulations for the $\gamma$  moderate-$\xi$ moderate scenario under C-asym contamination.}
\centering

\begin{tabular}{lrrrrrr}\hline
($\times 10$)  & bias($\hat{\mu}$) & sd($\hat{\mu})$ & $\sqrt{\mbox{mse}(\hat{\mu})}$ & bias($\hat{\sigma}$) & sd($\hat{\sigma})$ & $\sqrt{\mbox{mse}(\hat{\sigma})}$ \\ 
  \hline
IPW(X) & -8.835 & 1.506 & 8.962 & 15.903 & 1.511 & 15.974 \\ 
  AIPW(X,X) & -8.864 & 1.266 & 8.954 & 15.941 & 1.308 & 15.994 \\ 
  AIPW(X,XV) & -8.858 & 1.254 & 8.946 & 15.947 & 1.297 & 15.999 \\ 
  OR(X) & -8.873 & 1.292 & 8.967 & 17.684 & 1.023 & 17.714 \\ 
  OR(XV) & -8.866 & 1.278 & 8.958 & 17.774 & 1.008 & 17.803 \\ 
  RIPW(X) & 3.583 & 1.702 & 3.966 & -2.370 & 2.051 & 3.134 \\ 
  RAIPW(X,X) & 0.182 & 1.477 & 1.487 & -1.233 & 1.392 & 1.859 \\ 
  RAIPW(X,XV) & 0.165 & 1.452 & 1.460 & -1.207 & 1.357 & 1.815 \\ 
  ROR(X) & 0.017 & 1.298 & 1.297 & 0.295 & 1.021 & 1.062 \\ 
  ROR(XV) & 0.026 & 1.274 & 1.274 & 0.194 & 0.980 & 0.999 \\ 
   \hline
   IPW($X_{\_}$) & -7.618 & 1.443 & 7.754 & 15.815 & 1.451 & 15.881 \\ 
  AIPW($X_{\_}, XV$) & -8.865 & 1.261 & 8.954 & 16.153 & 1.292 & 16.204 \\ 
  AIPW($(X,X_{\_}V$) & -8.845 & 1.261 & 8.934 & 15.939 & 1.297 & 15.991 \\ 
  AIPW($X_{\_},X_{\_}V$) & -8.088 & 1.247 & 8.184 & 15.947 & 1.283 & 15.999 \\ 
  OR($X_{\_}$) & -8.091 & 1.264 & 8.189 & 17.559 & 1.003 & 17.588 \\ 
  RIPW($X_{\_}$) & 4.834 & 1.656 & 5.109 & -2.911 & 1.898 & 3.474 \\ 
  RAIPW($X_{\_}, XV$) & 0.113 & 1.447 & 1.451 & -1.193 & 1.330 & 1.786 \\ 
  RAIPW($X,X_{\_}V$) & 0.244 & 1.480 & 1.499 & -1.256 & 1.459 & 1.924 \\ 
  RAIPW($X_{\_},X_{\_}V$) & 1.034 & 1.452 & 1.782 & -1.616 & 1.386 & 2.128 \\ 
  ROR($X_{\_}V$) & 0.853 & 1.291 & 1.547 & 0.041 & 1.040 & 1.041 \\ \hline
\end{tabular} 
\end{table}

\begin{figure}
\centering
\includegraphics[width=\textwidth]{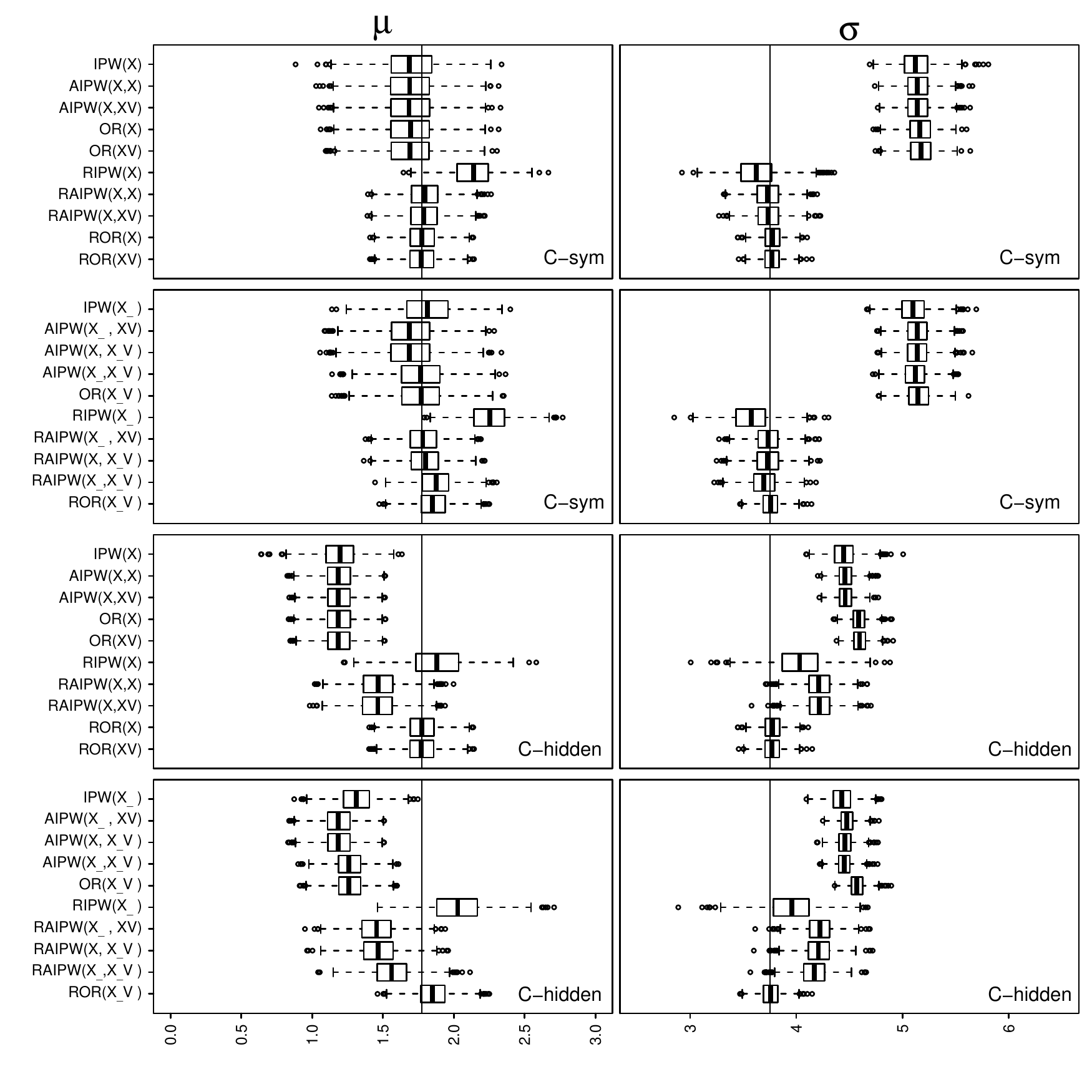} 
\caption{Estimates of $\mu$  (left) and $\sigma$ (right) for the $\gamma$  moderate-$\xi$ moderate scenario under the C-asym and C-hidden contamination.  The vertical lines represent the true underlying values.}
\label{Contam3ContamXiModerateGammaModerateFig}
\end{figure}

Results for the C-sym contamination are summarized in the top half of Figure~\ref{Contam3ContamXiModerateGammaModerateFig}, which shows the same patterns as C-asym in Figure~\ref{CleanContam2XiModerateGammaModerateFig}, with a major difference that the biases of RAIPW(X) and RAIPW(XV) in the estimation of $\sigma$ have now disappeared. Finally, the C-hidden configuration, whose results are summarized in the bottom half of Figure~\ref{Contam3ContamXiModerateGammaModerateFig}, confirms the expectation that this contamination scheme is most challenging. Here RAIPW is 
clearly biased for both $\mu$ and $\sigma$, and ROR behaves best both in terms of bias and efficiency. This is due to the fact that ROR only considers the conditional distribution of $Z_{2i}$ given $Z_{1i}$, through the outcome regression model, where the contamination is most visible, while RAIPW considers this conditional distribution but also the marginal of $Z_{2i}$ where the contamination is hidden. RMSE tables provided in Appendix~\ref{simulationmoderatemoderate.section} confirm the visual impression of the figures.

\section{Application: BMI change}\label{bmi.sec}

\begin{figure}
\centering
\includegraphics[width=0.45\textwidth]{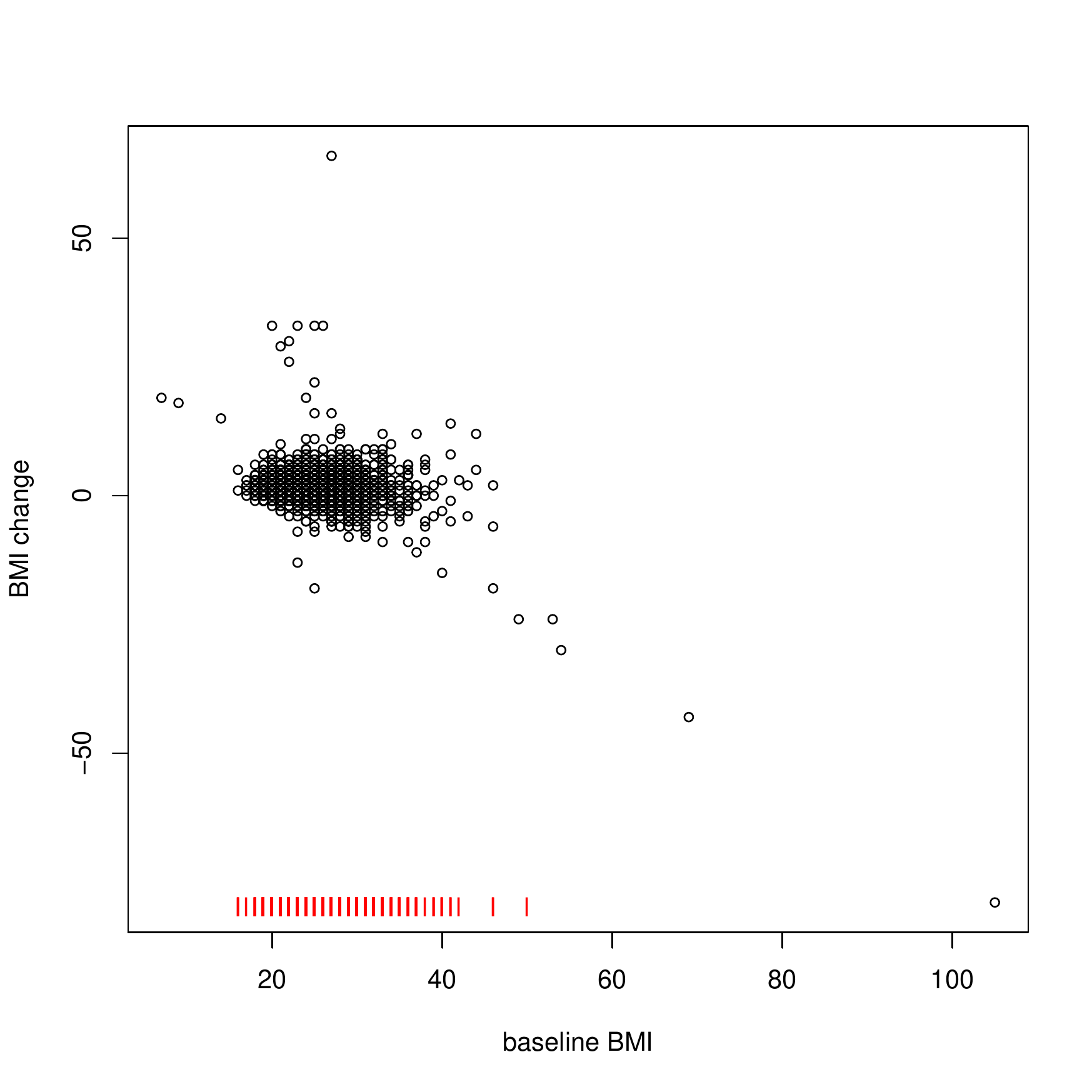}
\caption{BMI change after ten years versus BMI measured at 40 years of age for 
5553 men living in the county of V\"asterbotten in Sweden and born 1950-58. In red baseline BMI observed for the 2002 men not returning at follow up. }
\label{BMI.fig}
\end{figure}

We illustrate the methods presented in this paper with a population based 10 year follow up study of body mass index (kg/m$^2$, BMI).
 The analysis is performed on data from the Ume\aa\  SIMSAM Lab database
 \citep{simsam:2016}, which makes available record linkage information from
 several population based registers. In particular the database includes BMI
 data from an intervention program where all individuals living in the county of
 V\"asterbotten in Sweden and turning 40 and 50 years of age are called for 
 a health examination. Thanks to the Swedish individual identification number, this
 collected health data can be record linked to population wide health and
 administrative registers, which allows us to retrieve useful auxiliary
 information on individuals hospitalisation and socio-economy adding to self
 reported variables available from the intervention program.

We consider men born between 1950 and 1958 and observe
their BMI when they turn 40 years of age as well as at a 10 year follow
up. Figure \ref{BMI.fig} displays  a scatter plot of BMI change versus BMI at
baseline for the 5553 men who came back at the follow up examination when they
turned 50 out of the 7555 that were measured at baseline (40 years of
age). Extreme BMI values are observed (both at baseline and changes at follow up), giving an 
interesting case to illustrate the robust estimators introduced herein.
 
The set of baseline covariates used in the auxiliary models fitted are: (from the health
examinations) measured BMI, self reported health, and tobacco use; (from Statistics
Sweden registers) education level, number of children under 3 years of age, log 
earnings, parental benefits, sick leave benefits, unemployment benefits,
urban living; (from the hospitalisation register)
annual hospitalisation days.  These variables are used to explain dropout ($R_i$)
using a logistic regression and BMI change ($Z_2$) using linear
regression. Estimation of these auxiliary models is performed using maximum
likelihood or robust regression methods depending on the estimators
used.
More details on baseline covariates and results on the estimation of the auxiliary models are provided and discussed in Appendix~\ref{applicationsupp.section}. 


\begin{table}

\caption{\label{results.tab} Estimates and their standard errors (s.e.) for $\mu$ and $\sigma$ using the different  estimators defined in the article, and where "CC" stands for complete case sample moments.}
\centering

\begin{tabular}{cccccccc}\hline
	           & CC     &  IPW  & AIPW &  OR   & RIPW & RAIPW  & ROR   \\ \hline
 $\hat\mu$    & 1.43   & 1.36 & 1.38 & 1.41  & 1.34 & 1.34   & 1.34   \\ 
    s.e.      & 0.039  & 0.073 & 0.062 & 0.044  & 0.024 & 0.024   & 0.026  \\
 $\hat\sigma$ & 2.87   & 3.67 & 3.64 & 2.88  & 1.56 & 1.58   & 1.66   \\
   s.e.       & 1.524  & 0.700 & 0.665 & 0.264  & 0.025 & 0.025   & 0.024 \\ \hline 
	\end{tabular}
\end{table}

Table \ref{results.tab} displays estimated mean ($\hat\mu$) and standard deviation ($\hat\sigma$) in BMI change. The estimators used are the naive complete case (CC) sample mean and standard deviation (i.e., not
taking into account selective dropout and outlier contamination);
IPW, OR, and AIPW taking into
account selective dropout but not outlier contamination; and RIPW, ROR and RAIPW, taking into
account both selective dropout and contamination. 

\begin{figure} [htbp]
\centering
\includegraphics[width=0.6\textwidth]{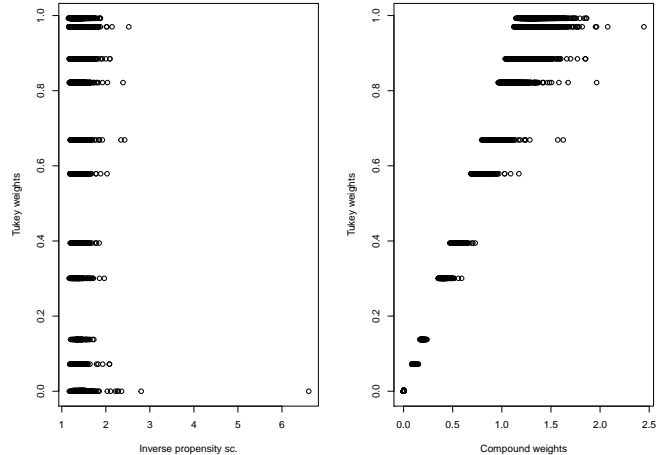}
\caption{Tukey's weights versus inverse of fitted propensity scores and compound weights (product of Tukey's weights and inverse propensity scores) used in RIPW.}
\label{BMIweights.fig}
\end{figure}

All three introduced robust versions (RIPW, RAIPW, ROR) give similar estimates of mean and
variance BMI change. Robust estimation seems to be here most relevant for the scale
parameter $\sigma$, where the robust versions are two BMI units smaller than the non-robust versions. This has also consequence for the estimation of standard errors for $\hat\mu$, IPW, OR and AIPW having standard errors 50\% larger than their robust counterparts. In summary, while robust estimation does not yield notably different results in mean BMI change, the results indicate a clear overestimation of the variability in BMI change if the non-robust estimators are used. Furthermore, correcting for selective dropout without taking into account contamination yields even larger overestimation of the variability of BMI change than using the naive estimator.

Taking into account selective dropout and contamination can both be seen as
weighting schemes as described in Remark \ref{weight.rem}. Thus, the propensity score weighting and the $\psi$ Tukey weighting, as well as the compound weights, are plotted in Figure \ref{BMIweights.fig} against each other. This plot highlights which individuals are downweighted due to outlying BMI change and overweighted due to selective dropout. 
 We observe that one observation has very high inverse propensity score and zero Tukey weight. This observation corresponds to the outlying individual with BMI at baseline close to 100 and large BMI decrease (see Figure \ref{BMI.fig}). Thus, while IPW and AIPW give it a large weight because it lies in a region with high probability of missingness, its outlying nature is noticed by the robust estimators which discard it from estimation. Its overweight in IPW and AIPW estimation is a contributing factor to their seemingly overestimation of $\sigma$ noticed above.

\section{Discussion}
In this paper we have studied semiparametric inference when outcome data is missing at random (ignorable given observed covariates) and the observed data is possibly contaminated by a nuisance process. We have proposed estimators which have bounded bias for an arbitrary large contamination. Many alternative estimators have been proposed in the literature \cite[for a review]{RotnitzkyStijn:15}. In order to obtain bounded influence function versions of those, the approach presented here can be followed. In particular, an interesting family of AIPW estimators are those which are bounded in the sense that they cannot produce an estimate outside the range of the observed outcomes \citep{Tan:2010,Grub:vand:2010}, in order to avoid the inverse probability weighting to have too drastic consequences. Such estimators are still not robust to contamination (unbounded influence function) and may therefore be robustified as proposed herein. Moreover, while we have focused on the location-scale parameter of the marginal distribution of the outcome of interest, the framework can readily be extended to other parameters, e.g., parameters of the conditional distribution of the outcome given some covariates, and to causal parameters defined using the potential outcome framework \citep{RR:74}. 

Finally, an interesting aspect of our results is that the auxiliary outcome regression model is not only useful to improve efficiency over IPW estimation but it is also useful for the robustness properties of the RAIPW and ROR estimators. This is the case in two ways, first because it allows us to relax assumptions (otherwise commonly made in robust estimation of location and scale) on the marginal distribution of the outcome, and second because it allows us to deal with contaminations which may be hidden when looking at the marginal distribution of the outcome, but become apparent in the conditional distribution of the outcome given the covariates.

\section*{Acknowledgements}
We are grateful to Ingeborg Waernbaum and Chris Field for comments that have improved the paper. We acknowledge the support of the Marianne and Marcus Wallenberg Foundation, the CRoNoS COST Action IC1408 and the Einar and Per Wikstr\"om Foundation. The Ume{\aa} SIMSAM Lab data infrastructure used in the application section was developed with support from the Swedish Research Council and by strategic funds from Ume{\aa} University. The simulations were performed at University of Geneva
on the Baobab cluster.


\appendix


\section{Proposition~2}
\label{proof.section}
We give below proofs of consistency and asymptotic normality for the RAIPW estimator. Assumptions made in Section \ref{theory.sec} on the data generating mechanism hold throughout.
Let us use the simpler notation $\varphi_{i}(\mu,\sigma;\gamma,\xi)$
for \\
$\varphi_{RAIPW}(Z_{i},R_i;\beta,\gamma,\xi)$ given in (\ref{robusts.aipw.eq}),
where the dependence on the data is shown merely by the index $i$. For convenience, the estimator $(\hat\mu_{RAIPW},\hat\sigma_{RAIPW})'$ defined as solving (\ref{robusts.aipw.eq}) is instead (and equivalently when such a solution exists) defined as a minimum distance estimator (between the empirical moment and zero):
\begin{equation}
\label{min.problem}
(\hat\mu_{RAIPW},\hat\sigma_{RAIPW})'=\arg\min_{\mu,\sigma} \hat Q(\mu,\sigma;\hat\gamma,\hat\xi), 
\end{equation}
\noindent where $\hat Q(\mu,\sigma;\hat\gamma,\hat\xi)=n^{-2}(\sum_{i=1}^{n}\varphi_{i}(\mu,\sigma;\hat\gamma,\hat\xi))'\sum_{i=1}^{n}\varphi_{i}(\mu,\sigma;\hat\gamma,\hat\xi)$. This allows us to utilize some general asymptotic results given in \cite{newey:mcfadden:94} as specified below. Proposition \ref{raipw.prop} given earlier is a summary of the two following propositions (Proposition 2 (consistency) and Proposition 2 (asymptotic normality)).

\subsection{Consistency}

\noindent\textbf{Regularity conditions}
\begin{itemize}
\item [A.1)] i) $\plim_{n\rightarrow\infty}\hat\gamma=\gamma^*$, and ii) $\plim_{n\rightarrow\infty}\hat\xi=\xi^*$;
\item [A.2)] \begin{itemize} 
\item[i)] $\pi(Z_{1i};\gamma)$ is differentiable with respect to $\gamma$ on the open interval with its derivative continuous on the closed interval between $\tilde\gamma$ and $\gamma$, where $\tilde\gamma\in[\gamma^*,\hat\gamma]$;
\item[ii)] $h_1(Z_{1i};\beta,\xi)$ and $h_2(Z_{1i};\beta,\xi)$ are differentiable with respect to $\xi$ on the open interval with their derivatives continuous on the closed interval between $\tilde\xi$ and $\xi$, where $\tilde\xi\in[\xi^*,\hat\xi]$;
\end{itemize}
\item [A.3)] $\beta=(\mu,\sigma)'\in \cal B$ where $\cal B$ is compact, and the equality
\begin{align*}
&E 
\left\{\begin{array}{c}
 \psi_{c_{\mu}}\left( \frac{Z_{2i}- \mu}{\sigma}\right)-A  \\
 \psi_{c_{\sigma}}^2\left( \frac{Z_{2i}- \mu}{\sigma}
\right)- B 
\end{array}\right\}=0
\end{align*}
holds only for $(\mu,\sigma)'=(\mu_0,\sigma_0)'\in\cal B$;
\item[A.4)] $1\geq \pi(Z_{1i};\gamma^*)>\varepsilon$ for $\varepsilon>0$ with probability (wp) 1;
 \item[A.5)] $\psi_c\big( \sigma^{-1}
(Z_{2i}- \mu)\big)$ is continuous at each $(\mu,\sigma)'\in \cal B$ wp 1;
\item[A.6)] $h_1(Z_{1i};\beta,\xi^*)$ and $h_2(Z_{1i};\beta,\xi^*)$ are continuous at each $\beta\in \cal B$ wp 1;
\item [A.7)] 

$$\hspace*{-1cm} E\big\{\sup_{(\mu,\sigma)'\in \cal B} \big( \big|
\psi_{c_{\mu}}\big( \sigma^{-1} (Z_{2i}- \mu)\big)\big|^2\big)\big\} + 
E\big\{\sup_{(\mu,\sigma)'\in \cal B}\big(\big| \psi_{c_{\sigma}}^2\big( \sigma^{-1}
(Z_{2i}- \mu)\big)\big|^2\big)\big\} <\infty ;$$

\item[A.8)] 

\hspace*{-1cm}
$$E\big\{\sup_{(\mu,\sigma)'\in \cal B}\big(\big( h_1(Z_{1i};\beta,\xi^*)\big)^2
\big)\big\} + 
E\big\{\sup_{(\mu,\sigma)'\in \cal B}\big(\big( h_2(Z_{1i};\beta,\xi^*)\big)^2\big)\big\}<\infty.$$
\end{itemize}

Condition A.3) is an identification condition. Compactness of $\cal B$ may be considered restrictive and can be relaxed at the cost of other assumptions \citep{newey:mcfadden:94}. For Huber's $\psi$ function, compactness is, for instance, not necessary \citep{Huber:1964}, while for Tukey's $\psi$ function identification holds only locally.

\setcounter{proposition}{1}
\begin{proposition}[consistency]
Let either 
 $\pi(Z_{1i};\gamma)$ be correctly specified with $\gamma^*=\gamma_0$ and/or $h(Z_{1i};\beta,\xi)=(h_1(Z_{1i};\beta,\xi),h_2(Z_{1i};\beta,\xi))'$ be correctly specified with $\xi^*=\xi_0$. 
 Then, under regularity conditions A.1) to A.8) given above,
 $$\plim_{n\rightarrow\infty} (\hat\mu_{RAIPW},\hat\sigma_{RAIPW})'=(\mu_0,\sigma_0)'.$$
\end{proposition}
\begin{proof}
This consistency result holds if the general Theorem 2.1 in \cite{newey:mcfadden:94} can be applied. A central assumption of this theorem is the following uniform convergence $$\plim_{n\rightarrow\infty} \sup_{(\mu,\sigma)'\in \cal B}||\hat Q(\hat\gamma,\hat\xi)-Q_0(\gamma^*,\xi^*)||=0,$$ 
where we use the simplified notation $\hat Q(\hat\gamma,\hat\xi)$ for $\hat Q(\mu,\sigma;\hat\gamma,\hat\xi)$,
and where
 $Q_0(\gamma^*,\xi^*)=E(\varphi_{i}(\mu,\sigma;\gamma^*,\xi^*)'\varphi_{i}(\mu,\sigma;\gamma^*,\xi^*))$ and $||\cdot||$ denotes the Euclidean norm.

In order to show this uniform convergence result, consider the Taylor expansion (using Assumption A.2)) of 
$\hat Q(\mu,\sigma;\hat\gamma,\hat\xi)$ as a function of $(\hat\gamma,\hat\xi)$ around $(\gamma^*,\xi^*)$:
\begin{align*}
\hat Q(\hat\gamma,\hat\xi)=\hat Q(\gamma^*,\xi^*)+
(\hat\lambda-\tilde\lambda)'D(\tilde\lambda),
\end{align*}
where $D(\tilde\lambda)=\frac{\partial\hat Q}{\partial\lambda}(\tilde\lambda)$, $\lambda=(\gamma,\xi)'$, and $\tilde\lambda\in[\lambda^*,\hat\lambda]$. We can then write:
\begin{align*}
\hat Q(\hat\gamma,\hat\xi)-Q_0(\gamma^*,\xi^*)=\hat Q(\gamma^*,\xi^*)-Q_0(\gamma^*,\xi^*)+
(\hat\lambda-\tilde\lambda)'D(\tilde\lambda).
\end{align*}
Thus, we have,
\begin{eqnarray*}
\lefteqn{\sup_{(\mu,\sigma)'\in \cal B}||\hat
Q(\hat\gamma,\hat\xi)-Q_0(\gamma^*,\xi^*)||} \\
& \leq & \sup_{(\mu,\sigma)'\in \cal B}||\hat Q(\gamma^*,\xi^*)-Q_0(\gamma^*,\xi^*)||+
\sup_{(\mu,\sigma)'\in \cal B}||(\hat\lambda-\tilde\lambda)'D(\tilde\lambda)|| 
\\
& \leq & \sup_{(\mu,\sigma)'\in \cal B}||\hat Q(\gamma^*,\xi^*)-Q_0(\gamma^*,\xi^*)||+ 
\bar D\sup_{(\mu,\sigma)'\in \cal B}||(\hat\lambda-\tilde\lambda)'||, 
\end{eqnarray*}
where $\bar D<\infty$ is the the supremum over ${(\mu,\sigma)' \in \cal B}$ of all the elements of the vector $ D(\tilde\lambda)$. By Assumption A.1), $\plim ||\hat\lambda-\tilde\lambda||=0$. It remains to show that 
$\plim \sup_{(\mu,\sigma)'\in \cal B}||\hat Q(\gamma^*,\xi^*)-Q_0(\gamma^*,\xi^*)||=0$.
This is a consequence of Theorem 2.6 in \cite{newey:mcfadden:94}, whose assumptions, referred below by (NMi-NMiv), are now verified. 

Assumption (NMi) says that $E(\varphi_{i}(\mu,\sigma;\gamma,\xi))=0$ should hold only for $(\mu,\sigma)'=(\mu_0,\sigma_0)'$. This is the case by Assumption A.3), and because either $\gamma^*=\gamma_0$ and/or $\xi^*=\xi_0$, so that either $E \left\{ \frac{R_i} {\pi(Z_{1i};\gamma^*)}\mid Z_{1i} \right\}=1$, when $\gamma^*=\gamma_0$, or $h(Z_{1i};\beta,\xi^*)=E(m_{\psi}(Z_{2i};\mu,\sigma)\mid Z_{1i})$, when $\xi^*=\xi_0$. 

Assumption (NMii) holds by Assumption A.3). 
 
Condition (NMiii) says that $\varphi_{i}(\mu,\sigma;\gamma^*,\xi^*)$ is continuous at each $(\mu,\sigma)'\in B$ with probability one. This holds by Assumptions A.4), A.5) and A.6).

We now show that condition (NMiv) holds:  
\begin{align*}
E&\left\{\sup_{(\mu,\sigma)'\in \cal B} \left(\left\Vert\varphi_i(\mu,\sigma;\gamma,\xi)\right\Vert\right) \right\}\\
&=E\left\{\sup_{(\mu,\sigma)'\in \cal B}\left(\left| \frac{R_i}{\pi(Z_{1i};\gamma^*)}\big( \psi_{c_{\mu}}\big( \sigma^{-1}
(Z_{2i}- \mu)\big)-A\big)- \frac{R_i-\pi(Z_{1i};\gamma^*)}{\pi(Z_{1i};\gamma^*)}h_1(Z_{1i};\beta,\xi^*) \right|^2 \right.\right.\\
&\ \ \ \ \ +\left.\left.\left| \frac{R_i}{\pi(Z_{1i};\gamma^*)}(\psi_{c_{\sigma}}^2\big( \sigma^{-1} 
(Z_{2i}- \mu)\big)- B)- \frac{R_i-\pi(Z_{1i};\gamma^*)}{\pi(Z_{1i};\gamma^*)}h_2(Z_{1i};\beta,\xi^*) \right|^2\right)\right\} \\
&\leq 2E\left\{\frac{R_i}{\pi(Z_{1i};\gamma^*)^2}\left(\sup_{(\mu,\sigma)'\in \cal B}\left(\left|\big( \psi_{c_{\mu}}\big( \sigma^{-1}
(Z_{2i}- \mu)\big)-A\big)\right|^2\right) \right.\right. \\
&\ \ \ \ \ \left.\left. +\sup_{(\mu,\sigma)'\in \cal B}\left(\left|\big( \psi_{c_{\sigma}}^2\big( \sigma^{-1}
(Z_{2i}- \mu)\big)-B\big)\right|^2 \right)\right)   \right\} \\
&\ \ \ \ \  + 2E\left\{ \left(\frac{R_i-\pi(Z_{1i};\gamma^*)}{\pi(Z_{1i};\gamma^*)}\right)^2 \left(\sup_{(\mu,\sigma)'\in
\cal B}\left(\big( h_1(Z_{1i};\beta,\xi^*)\big)^2\right) \right. \right. \\
& \ \ \ \ \  + 
\left. \left. \sup_{(\mu,\sigma)'\in \cal B}
\left(\big( h_2(Z_{1i};\beta,\xi^*)\big)^2 \right)  \right) \right\} \\
&=2E\left\{\frac{\pi(Z_{1i};\gamma_0)}{\pi(Z_{1i};\gamma^*)^2}\right\}\left(E\left\{\sup_{(\mu,\sigma)'\in \cal B}\left(\left|\big( \psi_{c_{\mu}}\big( \sigma^{-1}
(Z_{2i}- \mu)\big)-A\big)\right|^2\right)\right\} \right. \\
&\ \ \ \ \  \left.+E\left\{\sup_{(\mu,\sigma)'\in \cal B}\left(\left|\big( \psi_{c_{\sigma}}^2\big( \sigma^{-1}
(Z_{2i}- \mu)\big)-B\big)\right|^2 \right) \right\} \right)   \\
&\ \ \ \ \  + 2E\left\{\frac{(\pi(Z_{1i};\gamma_0)-\pi(Z_{1i};\gamma^*))^2}{\pi(Z_{1i};\gamma^*)^2}\right\}\left(E\left\{\sup_{(\mu,\sigma)'\in \cal B}\left(\big( h_1(Z_{1i};\beta,\xi^*)\big)^2\right)\right\} \right.\\
&\ \ \ \ \  \left.+
E\left\{\sup_{(\mu,\sigma)'\in \cal B}
\left(\big( h_2(Z_{1i};\beta,\xi^*)\big)^2\right)\right\}\right)  
<\infty,
\end{align*}
where we have used $(a+b)^2\leq 2(a^2+b^2)$, $\sup f+g\leq \sup f + \sup g$ for
$f,g$ positive, that $Z_{2i}$ and and that $R_i$ are independent conditional on
$Z_{1i}$. The last inequality holds   
by Assumptions A.4), A.7) and A.8).        
Finally, noting that (NMi-NMiv) also ensures that the remaining conditions of Theorem 2.1  hold \citep[Theorem 2.6]{newey:mcfadden:94} completes the proof.             
\end{proof}

\subsection{Asymptotic normality}
\noindent\textbf{Regularity conditions}

Let $\tilde\varphi_i(\theta)$ be the vector stacking $\varphi_i(\mu,\sigma)$, $m_\gamma(R_i,Z_{1i};\gamma)$ and $m_\xi(Z_i;\xi)$, with $\theta=(\mu,\sigma,\gamma',\xi')'$.
\begin{itemize}
	\item[A.9)] 
	$\theta_0\in \mbox{interior of }\cal T$, where $\theta_0=(\mu_0,\sigma_0,{\gamma^*}',{\xi^*}')'$ ;
		\item[A.10)] $E\big\{\tilde\varphi_i(\theta_0)\big\}=0$ and $E\big\{||\tilde\varphi_i(\theta_0)||^2\big\}$ is finite;
		\item[A.11)]
 $\tilde\varphi_i(\theta)$ is continuously differentiable in a neighborhood ${\cal N}$ of $\theta_0$;

\item[A.12)]

$E\{\sup_{\theta\in {\cal N}} ||\nabla_{\theta}\tilde\varphi_i(\theta)||\}<\infty$;

\item[A.13)] $E\big\{\nabla_{\theta} \tilde\varphi_i(\theta)\big\}'E\big\{\nabla_{\theta} \tilde\varphi_i(\theta)\big\}$ is nonsingular.

\end{itemize}

\setcounter{proposition}{1}
\begin{proposition}[asymptotic normality]
Let $\plim (\hat\mu_{RAIPW},\hat\sigma_{RAIPW})'=(\mu_0,\sigma_0)'$ and A.1) hold. Then, under regularity conditions A.9)-A.13) given above, $(\hat\mu_{RAIPW},\newline \hat\sigma_{RAIPW})'$ has the following asymptotic multivariate normal distribution as $n\rightarrow\infty$
$$\sqrt n \Big( (\hat\mu_{RAIPW},\hat\sigma_{RAIPW})'-(\mu_0,\sigma_0)'\Big )\overset{d}\rightarrow N\big(0,E\big\{IF_{RAIPW}(IF_{RAIPW})'\}\big),$$
where $IF_{RAIPW}(Z_{i},R_i; \beta)$ is given by (\ref{ifaipw.eq}).
\end{proposition}
\begin{proof}
This result is obtained by a direct application of Theorem 6.1 in \cite{newey:mcfadden:94}.	
\end{proof}
The latter result assumes that $\gamma$ and $\xi$ are estimated through M-estimators, and regularity conditions on the moment conditions $m_\gamma$ and $m_\xi$ are made. As pointed out in \citet[p.~2178]{newey:mcfadden:94}, the same result may be obtained for the more general situations where we have asymptotically linear estimators for $\gamma$ and $\xi$.

\section{Propositions 1 and 3}\label{prop1and3.sec}

Regularity conditions for Proposition \ref{ripw.prop} are as follows: for consistency, A.1)i), A.2)i), A.3), A.4), A.5) and A.7); for asymptotic normality, A.9)-A.13) where $\tilde\varphi_i(\theta)$ is the vector stacking $\varphi_i(\mu,\sigma)$ and $m_\gamma(R_i,Z_{1i};\gamma)$, and with $\theta=(\mu,\sigma,\gamma')'$.

Regularity conditions for Proposition \ref{ror.prop} are as follows: for consistency, A.1)ii), A.2)ii), A.3), A.6), and A.8); for asymptotic normality,  A.9)-A.13) where $\tilde\varphi_i(\theta)$ is the vector stacking $\varphi_i(\mu,\sigma)$ and $m_\xi(Z_i;\xi)$, and with $\theta=(\mu,\sigma,\xi')'$.

Proofs are similar to the above and are omitted.

\section{Implementation details}
\label{implementation.section}

We use the free software R  for our implementation. To obtain
RAIPW (given $\hat{\gamma}$ and $\hat{\xi}$),
instead of solving the system of equations
(\ref{robusts.aipw.eq}), we minimize
\eqref{min.problem} numerically with the function \texttt{optim}. We proceed
similarly for RIPW.

We use the set up of Example \ref{example.sls}, and
the conditional expectations 
$$E\Big( \psi_{c_{\mu}}(\sigma^{-1} (\tilde
h(Z_{1i};\xi_1)+\xi_2\nu- \mu)\big) | Z_{1i} \Big) \mbox{ and } E\Big(
\psi_{c_{\sigma}}^2(\sigma^{-1} 
(\tilde h(Z_{1i};\xi_1)+\xi_2\nu- \mu)\big) | Z_{1i} \Big)$$
defining the working model for RAIPW 
are obtained by numerical integration (function \texttt{integrate}), and A and B
by Monte Carlo simulations.

\subsection{Robust logistic regression}\label{roblogistic.sec}

The robust logistic regression estimator we consider to fit
model~\eqref{estim.gamma} 
is the
proposal by \cite{Cant:Ronc:2001} for generalized linear model (GLM). 
In the
case of a logistic model
$$ \log\left(\frac{\pi(Z_{1i},\gamma)}{1-\pi(Z_{1i},\gamma)} \right) =  Z_{1i}'
\gamma $$ 
it solves
\begin{equation}
\label{robGLM.estim}
\sum_{i=1}^n \Big[ \psi_c^L(r_i) w(Z_{1i}) \frac{1}{\sqrt{v_{\mu_i}}} 
\mu_i^{\prime}  - a(\gamma) \Big] =  0, 
\end{equation}
where $\mu_i = E(R_i | Z_{1i}) = \pi(Z_{1i},\gamma)$, $Var(R_i | Z_{1i}) = v_{\mu_i} =
\pi(Z_{1i},\gamma) (1-\pi(Z_{1i},\gamma))$, $r_i=
\frac{(R_i-\mu_i)}{\sqrt{v_{\mu_i}}}$ and $\mu_i^{\prime} = \partial \mu_i /
\partial \gamma$. The constant 
$$a(\gamma) = \frac{1}{n} \sum_{i=1}^n
E[\psi_c^L(r_i)] w(Z_{1i}) / \sqrt{v_{\mu_i}} \ \mu_i^{\prime}$$ 
is a correction
term to ensure   
Fisher consistency.

Estimator~\eqref{robGLM.estim} is implemented in the \texttt{R} function
\texttt{glmrob} in package \texttt{robustbase}, where $\psi_c^L(r_i)$ is
the Huber function. In our simulations and application we use the default
settings, namely $c=1.35$ and equal weights $w(Z_{1i})=1$ on the design.

\subsection{Robust regression}
\label{robreg.section}

For the robust fit of the auxiliary regression model, we consider $m_{\xi}(Z_i; \xi)$ in the estimating
equations~\eqref{estim.xi} as the joint M-estimator
of regression and scale \cite[Sec.~4.4.3]{Maro:Mart:Yoha:2006}
defined by solving  
\begin{equation}
\label{Mestim}
\sum_{i=1}^n \left(
\begin{array}{c}
\xi_2 \ \psi_{c_1}^R\left( \frac{Z_{2i} - Z_{1i}^T \xi_{1}}{\xi_2} \right) Z_{1i} \\
\xi_2^2 \left(  \psi_{c_2}^{HP2} \left( \frac{Z_{2i} - Z_{1i}^T \xi_{1}}{\xi_2}
\right)  \right)^2 -  \xi_2^2  \ a(\xi)
\end{array}
\right)= 0.
\end{equation}

In this work, we consider the Tukey function for $\psi_{c_1}^R$, because of its improved breakdown properties over the Huber
version. In addition, the redescending nature of the Tukey function also protect
against leverage points (outliers in the design space). The price to pay when
using a redescending estimator is the fact that the resulting estimating
equations admit more than one minimum. Careful implementation is therefore
required, in 
particular regarding the starting point of the algorithm. The function $
\psi_{c_2}^{HP2}$ is the Huber function, and $a(\xi) = E \left(   \left(
\psi_{c_2}^{HP2} \left( \frac{Z_{2i}- Z_{1i}^T \xi_{1}}{\xi_2} \right)
\right)^2 \right)$ is a consistency correction term. It is calibrated under the
Gaussian assumption, in which case it equals $2 \Phi(c_2) -1 - 2 c_2 \phi(c_2) +
2 c_2^2(1-\Phi(c_2))$, where $\phi$ and $\Phi$ are the density and cumulative
distribution function of a ${\mathcal N}(0,1)$ distribution, respectively. The
default value for $c_2$ (aiming at $95\%$ efficiency for the estimator of the
scale parameter) is  $1.345$.

\subsection{Standard errors of $\hat{\beta}$}
In what follows, we give the expression of the standard errors of $\hat{\beta}$
for the RAIPW defined by
\eqref{robusts.aipw.eq}, based on $\hat{\gamma}$ and
$\hat{\xi}$. The standard errors of the RIPW and ROR estimators can be derived including the straightfoward simplification
and are therefore omitted.

Let $\theta = (\beta^\prime,\gamma^\prime,\xi^\prime)^\prime$ be the vector of
all the parameters. Our proposal jointly solves  
\begin{equation}
\sum_{i=1}^n \Psi(Z_i, R_i, \theta) = \sum_{i=1}^n 
\left( \begin{array}{ccc}
 \varphi_{RAIPW}(Z_i,R_i; \beta, \gamma, \xi) \\
 m_\gamma (R_i,Z_{1i};\gamma)\\
 R_i m_\xi (Z_{i};\xi)
\end{array} \right) =0.
\end{equation}

An alternative expression for the asymptotic variance \citep{Stef:Boos:2002}  is  given by $V_n(\theta) = A_n^{-1}(\theta) B_n(\theta)
A_n^{-1}(\theta)$, where 
$$A_n(\theta) = - \frac{1}{n} \sum_{i=1}^n  \frac{\partial
\Psi(Z_i, R_i, \theta)}{\partial \theta} \ \ \mbox{and} \ \ B_n(\theta) = \frac{1}{n}  
\sum_{i=1}^n \Psi(Z_i, R_i, \theta) \Psi(Z_i, R_i, \theta)^\prime,$$
with
\begin{equation}
\label{An}
A_n(\theta) =  \left(
\begin{array}{c|cc}
\frac{1}{n} \sum_{i=1}^n \frac{\partial \varphi_{RAIPW,i}}{\partial \beta} & \frac{1}{n} \sum_{i=1}^n\frac{\partial \varphi_{RAIPW,i}}{\partial \gamma} & \frac{1}{n} \sum_{i=1}^n \frac{\partial \varphi_{RAIPW,i}}{\partial \xi}  \\
\hline
0_{p \times 2} &  \frac{1}{n} \sum_{i=1}^n\frac{\partial m_{\gamma,i}}{\partial \gamma}  & 0_{p \times dim(\xi)} \\
0_{p \times 2} & 0_{p \times p} & \frac{1}{n} \sum_{i=1}^n R_i \frac{\partial m_{\xi,i}}{\partial \xi}  \\
\end{array} \right) 
\end{equation}

\begin{eqnarray*}
\lefteqn{B_n(\theta) =} \\ 
& & \left(
\begin{array}{c|cc}
 \frac{1}{n} \sum_{i=1}^n \varphi_{RAIPW,i} \varphi_{RAIPW,i}^\prime &  \frac{1}{n} \sum_{i=1}^n \varphi_{RAIPW,i} m_{\gamma,i}^\prime &  \frac{1}{n} \sum_{i=1}^n R_i \varphi_{RAIPW,i} m_{\xi,i}^\prime \\ 
\hline
 \frac{1}{n} \sum_{i=1}^n  m_{\gamma,i} \varphi_{RAIPW,i}^\prime  &  \frac{1}{n} \sum_{i=1}^n  m_{\gamma,i} m_{\gamma,i}^\prime &   \frac{1}{n} \sum_{i=1}^n m_{\gamma,i} m_{\xi,i}^\prime \\
 \frac{1}{n} \sum_{i=1}^n  R_i m_{\xi,i} \varphi_{RAIPW,i}^\prime &  \frac{1}{n} \sum_{i=1}^n R_i m_{\xi,i} m_{\gamma,i}^\prime &  \frac{1}{n} \sum_{i=1}^n R_i m_{\xi,i} m_{\xi,i}^\prime
\end{array} \right) 
\end{eqnarray*}
where $\varphi_{RAIPW,i} = \varphi_{RAIPW}(Z_i,R_i; \beta, \gamma, \xi)$, $ m_{\gamma,i} = m_\gamma (R_i,Z_{1i};\gamma) $ and $m_{\xi,i} = m_\xi (Z_{i};\xi)$.

We estimate each matrix by plugging-in $\hat{\theta}$.
We use this sandwich estimator (rather than the formula involving the influence function, see Proposition~\ref{raipw.prop}) as it has been shown to be more stable in finite samples.

\section{Simulation complements}
We give in this Section some additional details pertaining to our simulation
setting of Section~\ref{simsetting.section}.  

\subsection{Tuning constant details for robust methods}
\label{tuning.section}

To make them comparable, we have tuned the robust methods to
have approximately $95\%$ efficiency at the correctly specified models for clean
data across simulations. This amounts to 
\begin{description}
\item[for RAIPW ] choosing $c_{\mu}$ and $c_{\sigma}$ in
Equation~\eqref{robusts.aipw.eq} (with 
$c= 4.685$ and $c_2 = 1.345$ in \eqref{Mestim}, and $c=1.345$ in
\eqref{robGLM.estim} kept fixed) .
\item[for ROR] choosing separately two values of $c_1$  in Equation \eqref{Mestim} (with $c_2 = 1.345$) to produce two  estimates $\hat{\xi}^{\mu}$ and $\hat{\xi}^{\sigma}$ of $\xi$, to be plugged-in  the solution of \eqref{imp.eq}.
\end{description}
 
Direct tuning of $c_{\mu}$ and $c_{\sigma}$ in
Equation~\eqref{robustmu.ipw.eq} for the robust IPW
estimator is impaired by the large bias, compared to variance, observed for the
estimator due to the difficulty in computing the Fisher consistency correction
term.  We have therefore taken the same values
of $c_{\mu}$ and $c_{\sigma}$  as for the robust AIPW (with $c=1.345$ in
\eqref{robGLM.estim} kept fixed).  
Table~\ref{tuningconstants.table} gives the values 
of the tuning constants obtained for the six considered designs and for each method.
 
 \begin{table}
 \caption{  \label{tuningconstants.table}  Tuning constants used in the simulation study. The values of $\xi$ and
  $\gamma$ under each scenario are given in \eqref{csi.value} and \eqref{gamma.value}.   }
 \centering 
\begin{tabular}{lcccc}\hline 
 & \multicolumn{2}{c}{RAIPW} &
 \multicolumn{2}{c}{ROR} \\ 
 & $c_{\mu}$ & $c_{\sigma}$  & $c_{\mu}$ & $c_{\sigma}$ \\ 
\hline 
$\gamma$ strong, $\xi$ strong & 3.7 & 4.5 & 3.3 & 3.7 \\
$\gamma$ strong, $\xi$ moderate & 3.9 & 4.5 & 3.4 & 3.7 \\
$\gamma$ strong, $\xi$ no & 4 & 4.5 & 3.6 & 4.2\\
\hline 
$\gamma$ moderate, $\xi$ strong & 3.2 & 5.3 & 2.6 & 2.8 \\
$\gamma$ moderate, $\xi$ moderate  & 3.9 & 5.4 & 3 & 3.1\\
$\gamma$ moderate, $\xi$ no  & 4.2 & 5.3 & 3.4 & 3.6\\ \hline
\end{tabular}  
\end{table}

\subsection{Computation of $\beta = (\mu,\sigma)'$}
We need to deduce the values of  $\beta = (\mu,\sigma)'$ for the designs simulated in Section~\ref{sim.section}.
To ease the computation of $\mu$ and $\sigma^2$, we rearrange
model~\eqref{sim.model} as follows: 
$$Z_{2i} = \xi_{10} + \tau^\prime \tilde{\xi} +\xi_{13}X_{3i}+ \xi_{16}
V_{3i} + \epsilon_i ,$$ 
where $\tau = (X_{1}, V_{1}, X_{2},  V_{2})'$ and $\tilde{\xi} =
(\xi_{11},\xi_{14},\xi_{12},\xi_{14} )'$. 

The derivation proceeds by conditioning on $X_{3}$ and using the law of total
expectation and variance. We have
\begin{align*}
\mu & =  \xi_{10} +  E_{X_3}\left( E_{\tau|X_3}(\tau^\prime \tilde{\xi})\right) +
E_{X_3}(\xi_{13} X_{3i}) +  E_{X_3}\left( E_{V_3 | X_3} \left( \xi_{16} V_{3i}
\right) \right) \\
 & =  \xi_{10} + E_{X_3}\left(\tau_{X_3}^\prime \tilde{\xi} \right) + \xi_{13}
 E_{X_3}(X_{3}) +  \xi_{16} E_{X_3}\left( 0.75 X_{3i} + 0.25 (1-X_{3i})\right)  \\
 & =  \xi_{10} + 0.2 \tau_1^\prime \tilde{\xi} +  0.8 \tau_0^\prime  \tilde{\xi}
 + 0.2 (\xi_{13}+  0.5  \xi_{16} )+ 0.25 \xi_{16}, 
\end{align*}
and
\begin{eqnarray*}
\sigma^2 & =  & Var\left( \tau^\prime \tilde{\xi} + \xi_{13} X_{3i} + xi_{16} V_{3i}
\right) + \xi_2^2 \\
& = & Var_{X_3}\left(E_{(\tau|,V)X_3}( \tau^\prime \tilde{\xi} + \xi_{13} X_{3i}
+ \xi_{16} V_{3i})\right) + \\
& & E_{X_3}\left( Var_{(\tau,V)|X_{3i}} \left(  \tau^\prime \tilde{\xi} + \xi_{13} X_{3i} + \xi_{16} V_{3i}\right) \right)  + \xi_2^2 \\
& =&  Var_{X_3}\left( \tau_{X_3}^\prime \tilde{\xi}   + \xi_{13} X_{3i} + 0.5 \xi_{16}  X_{3i} \right) +   \\
 & & E_{X_3}\left( \tilde{\xi}^\prime \Sigma_{X_3}    \tilde{\xi} + \xi_{16}^2  (0.5  X_{3i}  + 0.25)  (0.75 - 0.5  X_{3i}  )  \right) + \xi_2^2  \\
& = & 0.8 \left( \tau_0^\prime \tilde{\xi} - E_{X_3}(\tau_{X_3}^\prime \tilde{\xi} +  X_{3i} + 0.5 \xi_{16}  X_{3i} ) \right)^2 + \\
& & 0.2 \left( \tau_1^\prime \tilde{\xi} +  X_{3i} +0.5 \xi_{16} - E_{X_3}(\tau_{X_3}^\prime \tilde{\xi} +  \xi_{13} X_{3i} + 0.5 \xi_{16}  X_{3i} \right)^2  +\\
& & 0.8\left( \tilde{\xi}^\prime \Sigma_0 \tilde{\xi} + 0.25 \cdot 0.75 \xi_{16}^2 \right) + 0.2\left( \tilde{\xi}^\prime \Sigma_1 \tilde{\xi} + 0.25 \cdot 0.75 \xi_{16}^2 \right)  + \xi_2^2 \\
& = & 0.8 \left( \tau_0^\prime \tilde{\xi} - 0.8 \tau_{0}^\prime \tilde{\xi} - 0.2(\tau_{1}^\prime \tilde{\xi}  + \xi_{13} + 0.5 \xi_{16} ) \right)^2 + \\
& & 0.2 \left( \tau_1^\prime \tilde{\xi} +  \xi_{13} +0.5 \xi_{16} - 0.8 \tau_{0}^\prime \tilde{\xi} - 0.2(\tau_{1}^\prime \tilde{\xi}  + \xi_{13} + 0.5 \xi_{16}  )\right)^2  +\\
& & 0.8\left( \tilde{\xi}^\prime \Sigma_0 \tilde{\xi} + 0.25\cdot  0.75 \xi_{16}^2 \right) + 0.2\left( \tilde{\xi}^\prime \Sigma_1 \tilde{\xi} + 0.25 \cdot 0.75 \xi_{16}^2 \right)  + \xi_2^2.
\end{eqnarray*}

\clearpage
\subsection{Supplementary results for the $\gamma$ moderate-$\xi$ moderate design}
\label{simulationmoderatemoderate.section}

\begin{table}
\caption{\label{CleanModerateModerateTable} Bias, standard deviation and root mean squared error (times 10) of the estimates across simulations for the $\gamma$  moderate-$\xi$ moderate scenario for clean data. } 

\centering 
\begin{tabular}{lrrrrrr}\hline
($\times 10$)  & bias($\hat{\mu}$) & sd($\hat{\mu})$ & $\sqrt{\mbox{mse}(\hat{\mu})}$ & bias($\hat{\sigma}$) & sd($\hat{\sigma})$ & $\sqrt{\mbox{mse}(\hat{\sigma})}$ \\ 
  \hline
IPW(X) & -0.116 & 1.513 & 1.516 & -0.043 & 1.725 & 1.725 \\ 
  AIPW(X,X) & -0.119 & 1.243 & 1.248 & -0.015 & 1.079 & 1.079 \\ 
  AIPW(X,XV) & -0.113 & 1.217 & 1.222 & -0.020 & 1.014 & 1.014 \\ 
  OR(X) & -0.119 & 1.238 & 1.243 & 0.006 & 1.024 & 1.023 \\ 
  OR(XV) & -0.113 & 1.214 & 1.219 & 0.003 & 0.983 & 0.982 \\ 
  RIPW(X) & 3.149 & 1.500 & 3.488 & -0.904 & 1.794 & 2.008 \\ 
  RAIPW(X,X) & -0.118 & 1.284 & 1.289 & -0.021 & 1.111 & 1.111 \\ 
  RAIPW(X,XV) & -0.103 & 1.258 & 1.262 & -0.027 & 1.058 & 1.058 \\ 
  ROR(X) & -0.113 & 1.253 & 1.258 & -0.010 & 1.060 & 1.059 \\ 
  ROR(XV) & -0.109 & 1.230 & 1.235 & -0.007 & 1.005 & 1.005 \\ 
   \hline
IPW($X_{\_}$) & 1.169 & 1.468 & 1.876 & -0.514 & 1.595 & 1.675 \\ 
  AIPW($X_{\_}, XV$) & -0.114 & 1.217 & 1.221 & -0.016 & 1.001 & 1.001 \\ 
  AIPW($X,X_{\_} V$) & -0.105 & 1.227 & 1.231 & -0.027 & 1.081 & 1.081 \\ 
  AIPW($X_{\_},X_{\_}V$) & 0.692 & 1.232 & 1.412 & -0.334 & 1.053 & 1.105 \\ 
  OR($X_{\_}$) & 0.692 & 1.230 & 1.411 & -0.299 & 1.020 & 1.063 \\ 
  RIPW($X_{\_}$) & 4.381 & 1.510 & 4.634 & -1.389 & 1.725 & 2.214 \\ 
  RAIPW($X_{\_}, XV$) & -0.102 & 1.258 & 1.262 & -0.026 & 1.053 & 1.053 \\ 
  RAIPW($X,X_{\_}V$) & -0.084 & 1.286 & 1.289 & -0.056 & 1.137 & 1.137 \\ 
  RAIPW($X_{\_},X_{\_}V$) & 0.698 & 1.292 & 1.468 & -0.354 & 1.123 & 1.177 \\ 
  ROR($X_{\_}V$) & 0.681 & 1.268 & 1.439 & -0.290 & 1.055 & 1.094 \\ \hline
\end{tabular}
\end{table}

\begin{table}
\caption{\label{Contam3ModerateModerateTable} Bias, standard deviation and root mean squared error (times 10) of the estimates across simulations for the $\gamma$  moderate-$\xi$ moderate scenario under C-sym contamination. }

\centering
\begin{tabular}{lrrrrrr}\hline
($\times 10$)  & bias($\hat{\mu}$) & sd($\hat{\mu})$ & $\sqrt{\mbox{mse}(\hat{\mu})}$ & bias($\hat{\sigma}$) & sd($\hat{\sigma})$ & $\sqrt{\mbox{mse}(\hat{\sigma})}$ \\ 
  \hline
IPW(X) & -0.806 & 2.153 & 2.298 & 13.808 & 1.659 & 13.907 \\ 
  AIPW(X,X) & -0.880 & 2.011 & 2.194 & 13.902 & 1.390 & 13.972 \\ 
  AIPW(X,XV) & -0.878 & 1.997 & 2.181 & 13.898 & 1.375 & 13.966 \\ 
  OR(X) & -0.854 & 1.989 & 2.164 & 14.100 & 1.340 & 14.164 \\ 
  OR(XV) & -0.859 & 1.975 & 2.153 & 14.181 & 1.327 & 14.243 \\ 
  RIPW(X) & 3.567 & 1.619 & 3.917 & -1.209 & 2.191 & 2.501 \\ 
  RAIPW(X,X) & 0.217 & 1.407 & 1.423 & -0.208 & 1.478 & 1.492 \\ 
  RAIPW(X,XV) & 0.165 & 1.375 & 1.384 & -0.178 & 1.427 & 1.437 \\ 
  ROR(X) & 0.032 & 1.266 & 1.265 & 0.239 & 1.017 & 1.044 \\ 
  ROR(XV) & 0.008 & 1.253 & 1.253 & 0.188 & 0.975 & 0.993 \\ 
     \hline
     IPW($X{\_}$) & 0.379 & 2.091 & 2.124 & 13.521 & 1.576 & 13.612 \\ 
  AIPW($X_{\_}, XV$) & -0.872 & 1.983 & 2.166 & 13.900 & 1.354 & 13.966 \\ 
  AIPW($(X,X_{\_}V$) & -0.863 & 2.002 & 2.179 & 13.871 & 1.385 & 13.940 \\ 
  AIPW($X_{\_},X_{\_}V$) & -0.113 & 1.985 & 1.987 & 13.678 & 1.358 & 13.745 \\ 
  OR($X_{\_}$) & -0.103 & 1.976 & 1.978 & 13.954 & 1.329 & 14.017 \\ 
  RIPW($X_{\_}$) & 4.766 & 1.592 & 5.025 & -1.732 & 2.047 & 2.681 \\ 
  RAIPW($X_{\_}, XV$) & 0.111 & 1.368 & 1.371 & -0.161 & 1.403 & 1.412 \\ 
  RAIPW($X,X_{\_}V$) & 0.225 & 1.385 & 1.402 & -0.236 & 1.469 & 1.487 \\ 
  RAIPW($X_{\_},X_{\_}V$) & 0.989 & 1.374 & 1.692 & -0.577 & 1.442 & 1.553 \\ 
  ROR($X_{\_}V$) & 0.799 & 1.274 & 1.504 & 0.023 & 1.015 & 1.014 \\ \hline
\end{tabular} 
\end{table}

\begin{table}
\caption{\label{ContamModerateModerateTable}  Bias, standard deviation and root mean squared error (times 10) of the estimates across simulations for the $\gamma$  moderate-$\xi$ moderate scenario under C-hidden contamination.  }

\centering
\begin{tabular}{lrrrrrr}\hline
($\times 10$)  & bias($\hat{\mu}$) & sd($\hat{\mu})$ & $\sqrt{\mbox{mse}(\hat{\mu})}$ & bias($\hat{\sigma}$) & sd($\hat{\sigma})$ & $\sqrt{\mbox{mse}(\hat{\sigma})}$ \\ 
  \hline
IPW & -5.799 & 1.461 & 5.980 & 7.004 & 1.292 & 7.122 \\ 
  AIPW-X & -5.868 & 1.192 & 5.987 & 7.101 & 0.878 & 7.155 \\ 
  AIPW-XV & -5.873 & 1.170 & 5.989 & 7.102 & 0.847 & 7.153 \\ 
  OR-X & -5.864 & 1.204 & 5.986 & 8.390 & 0.817 & 8.430 \\ 
  OR-XV & -5.875 & 1.184 & 5.993 & 8.451 & 0.798 & 8.489 \\ 
  RIPW & 1.053 & 2.116 & 2.363 & 2.786 & 2.617 & 3.821 \\ 
  RAIPW-X & -3.031 & 1.559 & 3.408 & 4.581 & 1.488 & 4.816 \\ 
  RAIPW-XV & -3.089 & 1.519 & 3.442 & 4.615 & 1.440 & 4.835 \\ 
  ROR-X & 0.025 & 1.268 & 1.267 & 0.254 & 1.023 & 1.053 \\ 
  ROR-XV & 0.004 & 1.255 & 1.254 & 0.197 & 0.979 & 0.999 \\ 
   \hline
   IPW($X_{\_}$) & -4.617 & 1.405 & 4.826 & 6.808 & 1.185 & 6.910 \\ 
  AIPW($X_{\_}, XV$) & -5.868 & 1.174 & 5.985 & 7.238 & 0.832 & 7.285 \\ 
  AIPW($(X,X_{\_}V$) & -5.858 & 1.177 & 5.975 & 7.073 & 0.860 & 7.126 \\ 
  AIPW($X_{\_},X_{\_}V$) & -5.110 & 1.179 & 5.244 & 6.985 & 0.839 & 7.035 \\ 
  OR($X_{\_}$) & -5.118 & 1.186 & 5.254 & 8.207 & 0.801 & 8.246 \\ 
  RIPW($X_{\_}$) & 2.548 & 2.078 & 3.287 & 1.968 & 2.573 & 3.238 \\ 
  RAIPW($X_{\_}, XV$) & -3.157 & 1.522 & 3.504 & 4.643 & 1.426 & 4.857 \\ 
  RAIPW($X,X_{\_}V$) & -3.013 & 1.531 & 3.379 & 4.544 & 1.493 & 4.783 \\ 
  RAIPW($X_{\_},X_{\_}V$) & -2.088 & 1.560 & 2.607 & 4.084 & 1.522 & 4.358 \\ 
  ROR($X_{\_}V$) & 0.789 & 1.277 & 1.500 & 0.042 & 1.021 & 1.021 \\ \hline
\end{tabular}
\end{table}

\clearpage
\subsection{ results for the other $\gamma$ - $\xi$ combinations}
\label{simulationothers.section}
\begin{center}
\begin{figure}
\centering
\includegraphics[width=\textwidth]{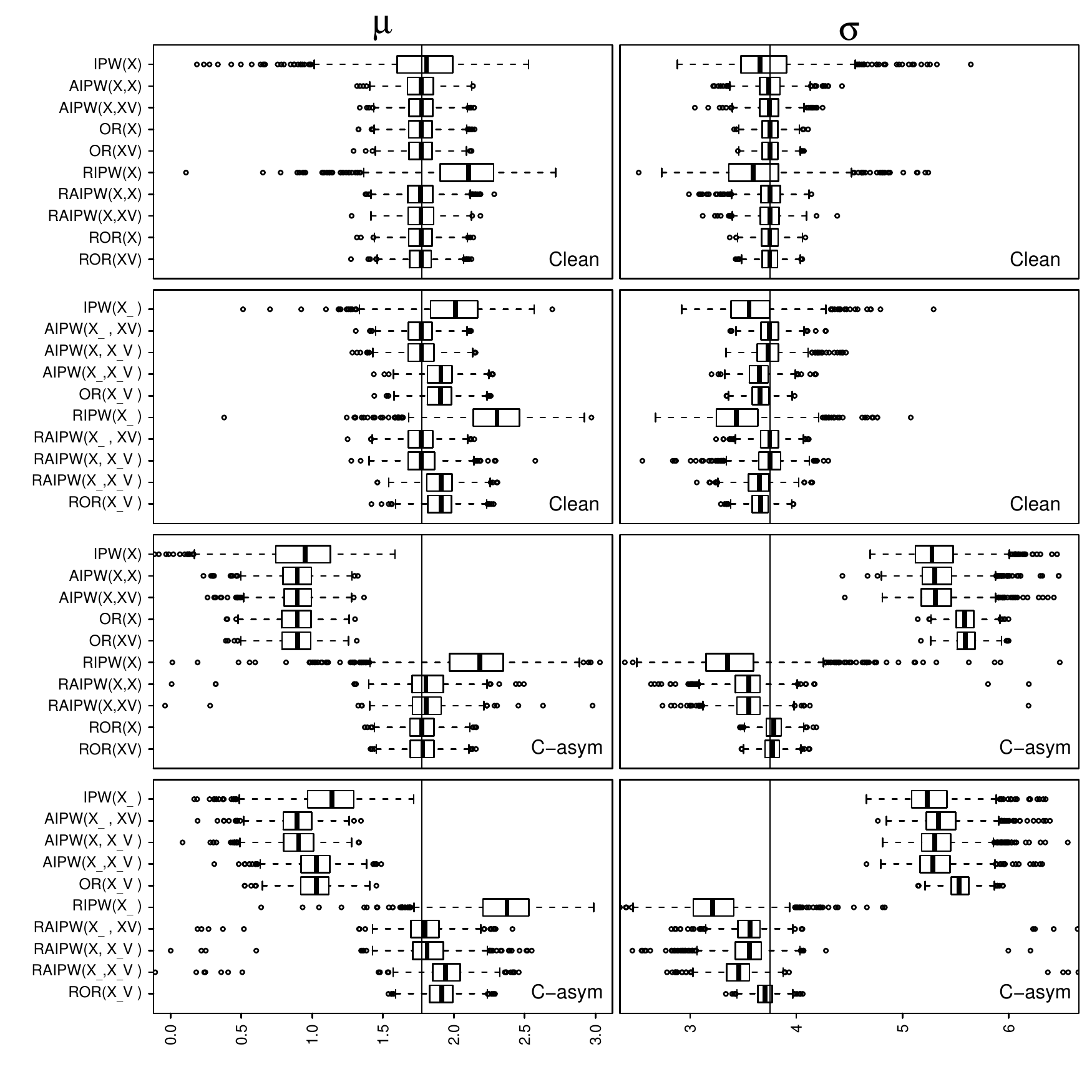} 
\caption{Estimates of $\mu$  (left) and $\sigma$ (right) for the $\gamma$  strong-$\xi$ moderate scenario for clean data and under the C-asym contamination.  The vertical lines represent the true underlying values.}
\label{CleanContam2XiModerateGammaStrongFig}
\end{figure}
\end{center}

\begin{center}
\begin{figure}
\centering
\includegraphics[width=\textwidth]{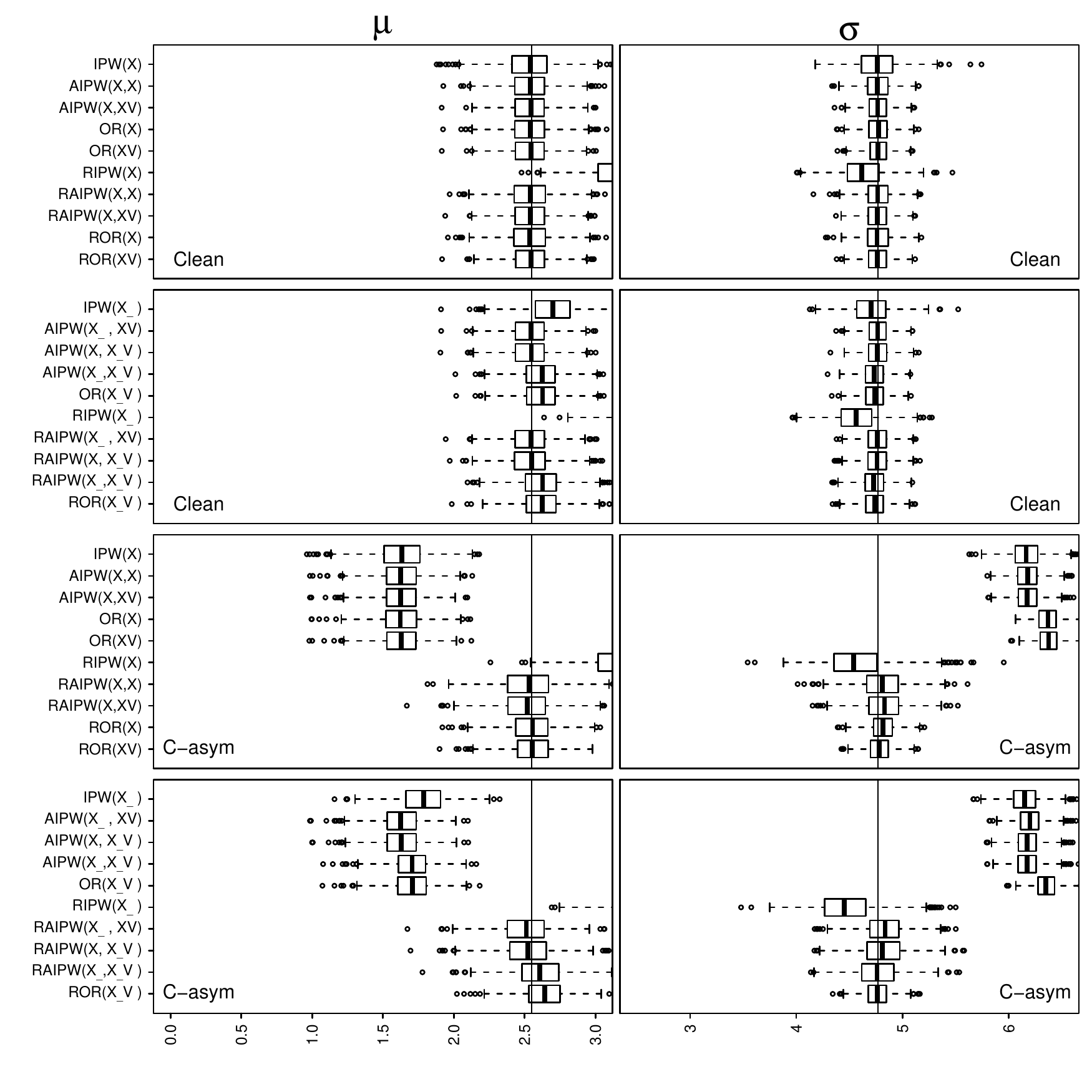} 
\caption{Estimates of $\mu$  (left) and $\sigma$ (right) for the $\gamma$  moderate-$\xi$ strong scenario for clean data and under the C-asym contamination.  The vertical lines represent the true underlying values.}
\label{CleanContam2XiStrongGammaModerateFig}
\end{figure}
\end{center}

\begin{center}
\begin{figure}
\centering
\includegraphics[width=\textwidth]{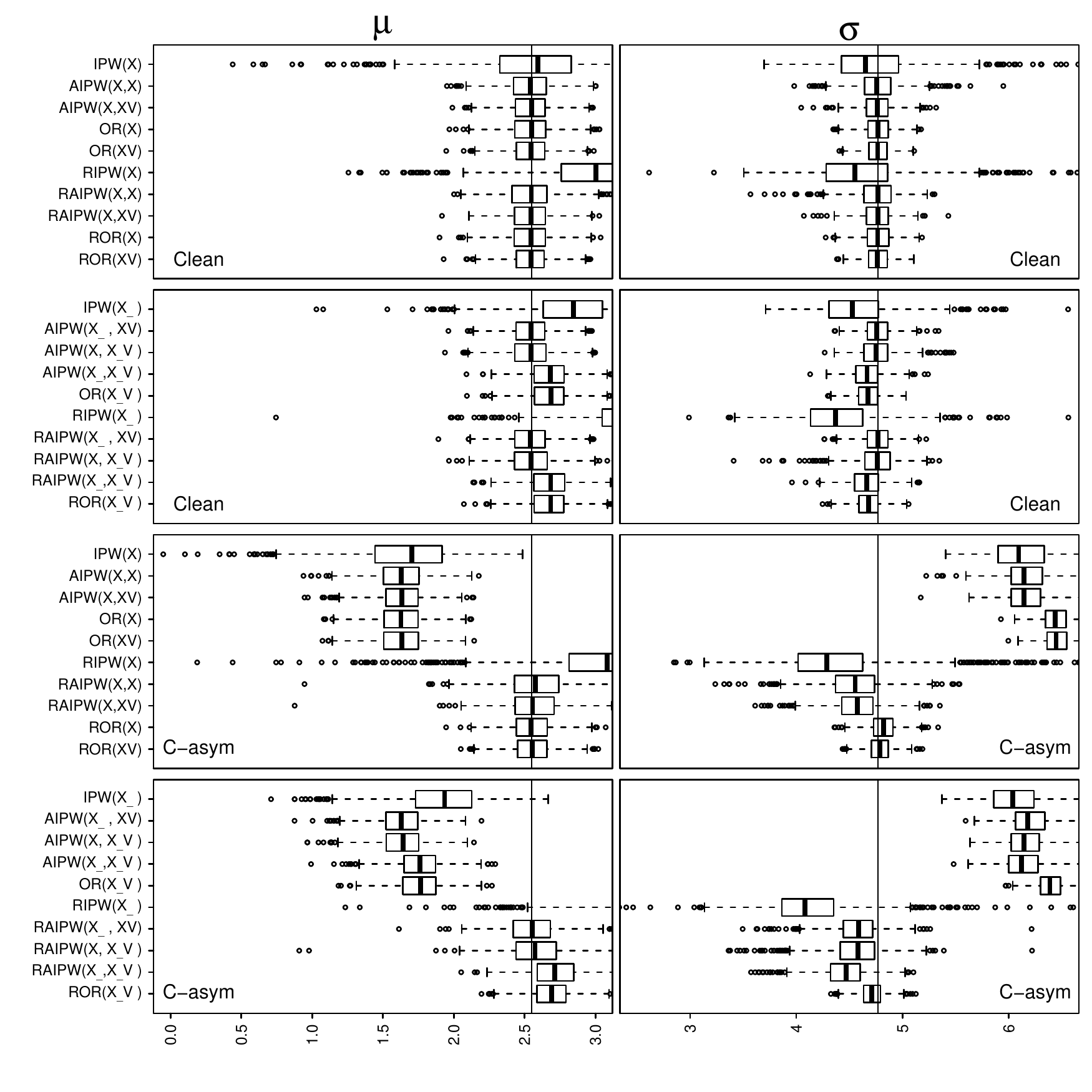} 
\caption{Estimates of $\mu$  (left) and $\sigma$ (right) for the $\gamma$  strong-$\xi$ strong scenario for clean data and under the C-asym contamination.  The vertical lines represent the true underlying values.}
\label{CleanContam2XiStrongGammaStrongFig}
\end{figure}
\end{center}

\begin{center}
\begin{figure}
\centering
\includegraphics[width=\textwidth]{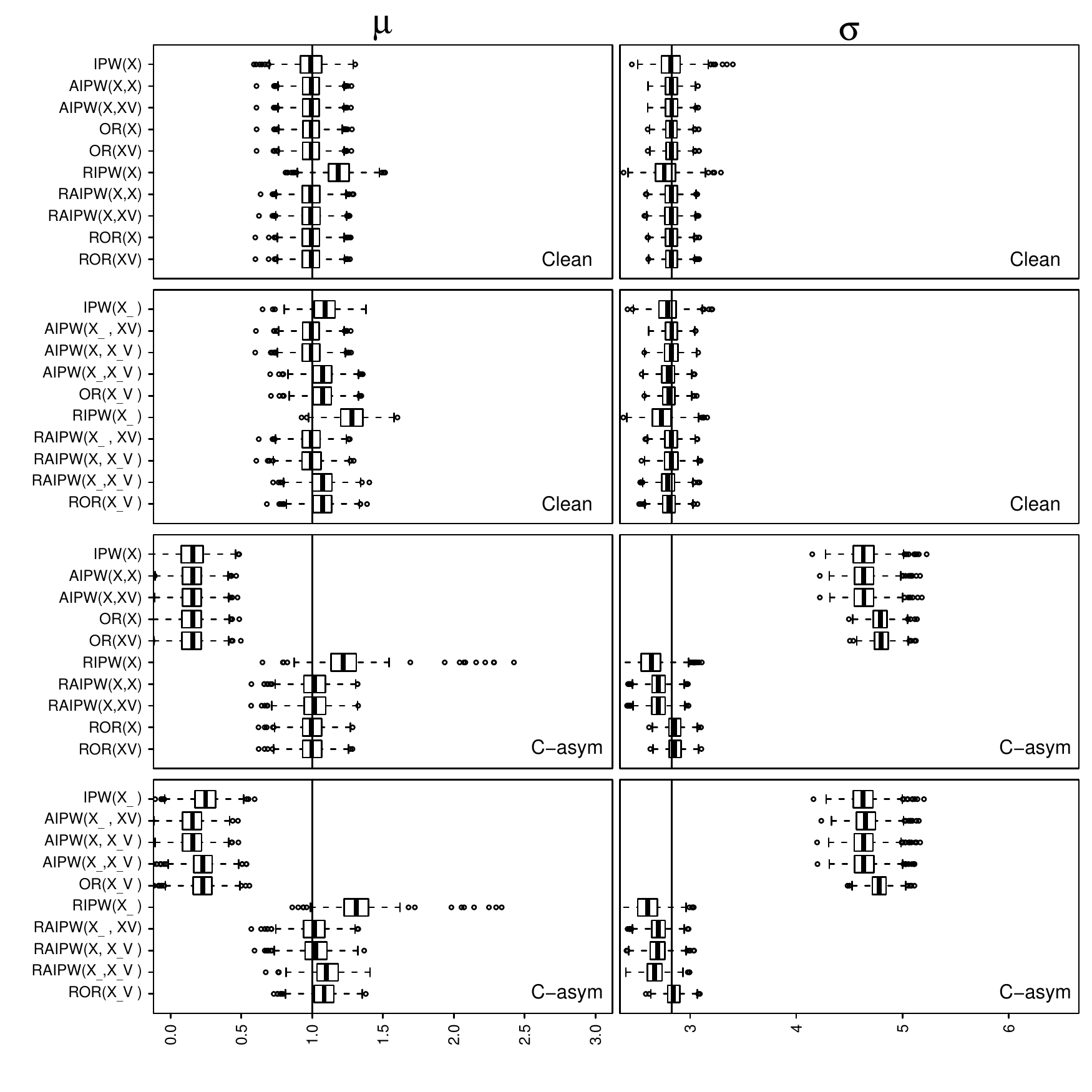} 
\caption{Estimates of $\mu$  (left) and $\sigma$ (right) for the $\gamma$  moderate-$\xi$ no scenario for clean data and under the C-asym contamination.  The vertical lines represent the true underlying values.}
\label{CleanContam2XiNoGammaModerateFig}
\end{figure}
\end{center}

\begin{center}
\begin{figure}
\centering
\includegraphics[width=\textwidth]{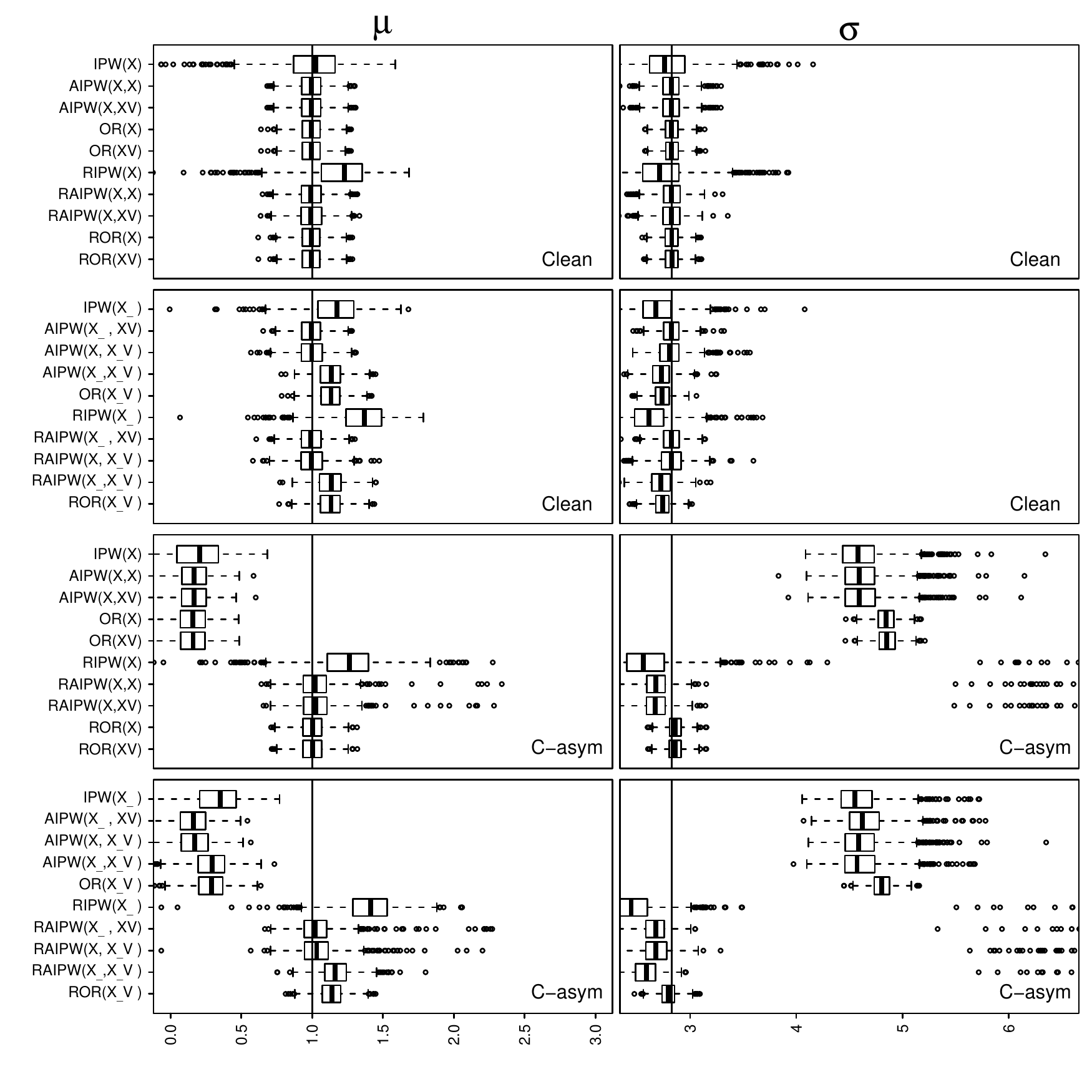} 
\caption{Estimates of $\mu$  (left) and $\sigma$ (right) for the $\gamma$ strong-$\xi$ no scenario for clean data and under the C-asym contamination.  The vertical lines represent the true underlying values.}
\label{CleanContam2XiNoGammaStrongFig}
\end{figure}
\end{center}

\begin{table}
\caption{Bias, standard deviation and root mean squared error (times 10) of the estimates across simulations for the $\gamma$  strong-$\xi$ moderate scenario for clean data.} 
\centering
\begin{tabular}{lrrrrrr} \hline
($\times 10$)  & bias($\hat{\mu}$) & sd($\hat{\mu})$ & $\sqrt{\mbox{mse}(\hat{\mu})}$ & bias($\hat{\sigma}$) & sd($\hat{\sigma})$ & $\sqrt{\mbox{mse}(\hat{\sigma})}$ \\ \hline
IPW(X) & -0.045 & 3.245 & 3.243 & -0.354 & 3.540 & 3.556 \\ 
  AIPW(X,X) & -0.101 & 1.315 & 1.318 & -0.052 & 1.514 & 1.514 \\ 
  AIPW(X,XV) & -0.085 & 1.269 & 1.271 & -0.076 & 1.348 & 1.349 \\ 
  OR(X) & -0.106 & 1.259 & 1.263 & 0.000 & 1.107 & 1.106 \\ 
  OR(XV) & -0.097 & 1.227 & 1.230 & -0.007 & 1.048 & 1.047 \\ 
  RIPW(X) & 2.856 & 3.179 & 4.272 & -1.170 & 3.870 & 4.041 \\ 
  RAIPW(X,X) & -0.090 & 1.415 & 1.417 & -0.059 & 1.545 & 1.546 \\ 
  RAIPW(X,XV) & -0.086 & 1.343 & 1.345 & -0.048 & 1.396 & 1.396 \\ 
  ROR(X) & -0.104 & 1.271 & 1.274 & -0.013 & 1.141 & 1.140 \\ 
  ROR(XV) & -0.092 & 1.235 & 1.238 & -0.014 & 1.069 & 1.069 \\ 
   \hline
   IPW($X_{\_}$) & 2.162 & 2.550 & 3.343 & -1.702 & 2.832 & 3.303 \\ 
  AIPW($X_{\_}, XV$) & -0.094 & 1.249 & 1.252 & -0.049 & 1.216 & 1.216 \\ 
  AIPW($(X,X_{\_}V$) & -0.053 & 1.336 & 1.336 & -0.126 & 1.590 & 1.594 \\ 
  AIPW($X_{\_},X_{\_}V$) & 1.296 & 1.277 & 1.819 & -1.020 & 1.359 & 1.699 \\ 
  OR($X_{\_}$) & 1.280 & 1.239 & 1.781 & -0.924 & 1.103 & 1.439 \\ 
  RIPW($X_{\_}$) & 5.081 & 2.636 & 5.723 & -2.942 & 3.209 & 4.352 \\ 
  RAIPW($X_{\_}, XV$) & -0.088 & 1.315 & 1.317 & -0.041 & 1.296 & 1.296 \\ 
  RAIPW($X,X_{\_}V$) & 0.009 & 1.461 & 1.460 & -0.145 & 1.762 & 1.767 \\ 
  RAIPW($X_{\_},X_{\_}V$) & 1.305 & 1.361 & 1.885 & -1.091 & 1.478 & 1.836 \\ 
  ROR($X_{\_}V$) & 1.274 & 1.267 & 1.796 & -0.921 & 1.141 & 1.465 \\ \hline
\end{tabular}
\end{table}

\begin{table}
\caption{Bias, standard deviation and root mean squared error (times 10) of the estimates across simulations for the $\gamma$  moderate-$\xi$ strong scenario for clean data.} 
\centering
\begin{tabular}{lrrrrrr}\hline
($\times 10$)  & bias($\hat{\mu}$) & sd($\hat{\mu})$ & $\sqrt{\mbox{mse}(\hat{\mu})}$ & bias($\hat{\sigma}$) & sd($\hat{\sigma})$ & $\sqrt{\mbox{mse}(\hat{\sigma})}$ \\ \hline
IPW(X) & -0.141 & 1.930 & 1.934 & -0.041 & 2.179 & 2.178 \\ 
  AIPW(X,X) & -0.145 & 1.596 & 1.602 & -0.007 & 1.401 & 1.400 \\ 
  AIPW(X,XV) & -0.133 & 1.520 & 1.525 & -0.016 & 1.212 & 1.211 \\ 
  OR(X) & -0.145 & 1.587 & 1.593 & 0.023 & 1.313 & 1.312 \\ 
  OR(XV) & -0.134 & 1.518 & 1.523 & 0.003 & 1.188 & 1.187 \\ 
  RIPW(X) & 6.032 & 1.938 & 6.336 & -1.411 & 2.234 & 2.641 \\ 
  RAIPW(X,X) & -0.147 & 1.715 & 1.720 & -0.024 & 1.455 & 1.454 \\ 
  RAIPW(X,XV) & -0.120 & 1.586 & 1.589 & -0.026 & 1.279 & 1.279 \\ 
  ROR(X) & -0.169 & 1.635 & 1.643 & -0.021 & 1.392 & 1.392 \\ 
  ROR(XV) & -0.128 & 1.538 & 1.543 & -0.013 & 1.215 & 1.215 \\ 
   \hline
  IPW($X_{\_}$) & 1.464 & 1.876 & 2.379 & -0.622 & 2.017 & 2.109 \\ 
  AIPW($X_{\_}, XV$) & -0.135 & 1.520 & 1.525 & -0.012 & 1.201 & 1.200 \\ 
  AIPW($(X,X_{\_}V$) & -0.125 & 1.528 & 1.532 & -0.022 & 1.271 & 1.271 \\ 
  AIPW($X_{\_},X_{\_}V$) & 0.672 & 1.530 & 1.670 & -0.331 & 1.247 & 1.289 \\ 
  OR($X_{\_}$) & 0.672 & 1.528 & 1.669 & -0.299 & 1.219 & 1.255 \\ 
  RIPW($X_{\_}$) & 7.509 & 1.976 & 7.764 & -1.993 & 2.161 & 2.939 \\ 
  RAIPW($X_{\_}, XV$) & -0.118 & 1.587 & 1.590 & -0.025 & 1.273 & 1.273 \\ 
  RAIPW($X,X_{\_}V$) & -0.101 & 1.625 & 1.627 & -0.060 & 1.352 & 1.352 \\ 
  RAIPW($X_{\_},X_{\_}V$) & 0.670 & 1.628 & 1.760 & -0.353 & 1.337 & 1.382 \\ 
  ROR($X_{\_}V$) & 0.658 & 1.581 & 1.712 & -0.293 & 1.263 & 1.296 \\ \hline
\end{tabular}
\end{table}

\begin{table}
\caption{Bias, standard deviation and root mean squared error (times 10) of the estimates across simulations for the $\gamma$  strong-$\xi$ strong scenario for clean data.} 
\centering
\begin{tabular}{lrrrrrr}\hline
($\times 10$)  & bias($\hat{\mu}$) & sd($\hat{\mu})$ & $\sqrt{\mbox{mse}(\hat{\mu})}$ & bias($\hat{\sigma}$) & sd($\hat{\sigma})$ & $\sqrt{\mbox{mse}(\hat{\sigma})}$ \\ \hline
IPW(X) & -0.064 & 4.109 & 4.107 & -0.408 & 4.464 & 4.480 \\ 
  AIPW(X,X) & -0.138 & 1.713 & 1.718 & -0.029 & 2.032 & 2.032 \\ 
  AIPW(X,XV) & -0.105 & 1.561 & 1.564 & -0.069 & 1.495 & 1.496 \\ 
  OR(X) & -0.136 & 1.624 & 1.629 & 0.024 & 1.425 & 1.424 \\ 
  OR(XV) & -0.118 & 1.527 & 1.531 & -0.007 & 1.241 & 1.241 \\ 
  RIPW(X) & 4.024 & 3.920 & 5.617 & -1.556 & 4.858 & 5.099 \\ 
  RAIPW(X,X) & -0.106 & 1.894 & 1.896 & -0.099 & 2.085 & 2.086 \\ 
  RAIPW(X,XV) & -0.112 & 1.639 & 1.642 & -0.043 & 1.589 & 1.588 \\ 
  ROR(X) & -0.154 & 1.652 & 1.659 & -0.006 & 1.484 & 1.483 \\ 
  ROR(XV) & -0.112 & 1.531 & 1.534 & -0.015 & 1.259 & 1.258 \\ 
   \hline
   IPW($X_{\_}$) & 2.695 & 3.233 & 4.208 & -2.076 & 3.559 & 4.119 \\ 
  AIPW($X_{\_}, XV$) & -0.114 & 1.545 & 1.548 & -0.043 & 1.382 & 1.382 \\ 
  AIPW($(X,X_{\_}V$) & -0.074 & 1.616 & 1.617 & -0.117 & 1.717 & 1.720 \\ 
  AIPW($X_{\_},X_{\_}V$) & 1.276 & 1.563 & 2.017 & -1.018 & 1.504 & 1.815 \\ 
  OR($X_{\_}$) & 1.259 & 1.532 & 1.982 & -0.920 & 1.283 & 1.579 \\ 
  RIPW($X_{\_}$) & 6.730 & 3.317 & 7.503 & -3.679 & 4.015 & 5.445 \\ 
  RAIPW($X_{\_}, XV$) & -0.112 & 1.614 & 1.618 & -0.035 & 1.489 & 1.488 \\ 
  RAIPW($X,X_{\_}V$) & -0.011 & 1.752 & 1.751 & -0.148 & 1.957 & 1.961 \\ 
  RAIPW($X_{\_},X_{\_}V$) & 1.274 & 1.653 & 2.086 & -1.078 & 1.662 & 1.980 \\ 
  ROR($X_{\_}V$) & 1.254 & 1.555 & 1.996 & -0.920 & 1.314 & 1.604 \\ \hline
\end{tabular}
\end{table}

\begin{table}
\caption{Bias, standard deviation and root mean squared error (times 10) of the estimates across simulations for the $\gamma$  moderate-$\xi$ no scenario for clean data.} 
\centering
\begin{tabular}{lrrrrrr}\hline
($\times 10$)  & bias($\hat{\mu}$) & sd($\hat{\mu})$ & $\sqrt{\mbox{mse}(\hat{\mu})}$ & bias($\hat{\sigma}$) & sd($\hat{\sigma})$ & $\sqrt{\mbox{mse}(\hat{\sigma})}$ \\ \hline
IPW(X) & -0.090 & 1.143 & 1.146 & -0.046 & 1.295 & 1.295 \\ 
  AIPW(X,X) & -0.092 & 0.947 & 0.951 & -0.026 & 0.837 & 0.837 \\ 
  AIPW(X,XV) & -0.092 & 0.947 & 0.951 & -0.025 & 0.837 & 0.837 \\ 
  OR(X) & -0.093 & 0.944 & 0.948 & -0.006 & 0.796 & 0.796 \\ 
  OR(XV) & -0.093 & 0.944 & 0.948 & 0.004 & 0.797 & 0.797 \\ 
  RIPW(X) & 1.867 & 1.138 & 2.187 & -0.603 & 1.356 & 1.483 \\ 
  RAIPW(X,X) & -0.093 & 0.982 & 0.986 & -0.026 & 0.866 & 0.866 \\ 
  RAIPW(X,XV) & -0.083 & 0.983 & 0.986 & -0.033 & 0.870 & 0.870 \\ 
  ROR(X) & -0.091 & 0.958 & 0.962 & -0.013 & 0.815 & 0.814 \\ 
  ROR(XV) & -0.090 & 0.957 & 0.961 & 0.003 & 0.814 & 0.813 \\ 
   \hline
   IPW($X_{\_}$) & 0.874 & 1.108 & 1.411 & -0.396 & 1.200 & 1.263 \\ 
  AIPW($X_{\_}, XV$) & -0.094 & 0.946 & 0.950 & -0.021 & 0.822 & 0.822 \\ 
  AIPW($(X,X_{\_}V$) & -0.085 & 0.960 & 0.963 & -0.033 & 0.908 & 0.909 \\ 
  AIPW($X_{\_},X_{\_}V$) & 0.712 & 0.969 & 1.202 & -0.327 & 0.878 & 0.936 \\ 
  OR($X_{\_}$) & 0.713 & 0.967 & 1.201 & -0.289 & 0.839 & 0.888 \\ 
  RIPW($X_{\_}$) & 2.802 & 1.138 & 3.024 & -0.967 & 1.304 & 1.623 \\ 
  RAIPW($X_{\_}, XV$) & -0.082 & 0.983 & 0.986 & -0.033 & 0.866 & 0.866 \\ 
  RAIPW($X,X_{\_}V$) & -0.064 & 1.010 & 1.012 & -0.054 & 0.952 & 0.953 \\ 
  RAIPW($X_{\_},X_{\_}V$) & 0.720 & 1.022 & 1.249 & -0.347 & 0.939 & 1.001 \\ 
  ROR($X_{\_}V$) & 0.704 & 0.997 & 1.220 & -0.274 & 0.864 & 0.906 \\ \hline
\end{tabular}
\end{table}

\begin{table}
\caption{Bias, standard deviation and root mean squared error (times 10) of the estimates across simulations for the $\gamma$  strong-$\xi$ no scenario for clean data.} 
\centering
\begin{tabular}{lrrrrrr}\hline
($\times 10$)  & bias($\hat{\mu}$) & sd($\hat{\mu})$ & $\sqrt{\mbox{mse}(\hat{\mu})}$ & bias($\hat{\sigma}$) & sd($\hat{\sigma})$ & $\sqrt{\mbox{mse}(\hat{\sigma})}$ \\ \hline
IPW(X) & -0.025 & 2.420 & 2.419 & -0.293 & 2.623 & 2.638 \\ 
  AIPW(X,X) & -0.064 & 1.012 & 1.014 & -0.085 & 1.216 & 1.219 \\ 
  AIPW(X,XV) & -0.065 & 1.013 & 1.014 & -0.084 & 1.215 & 1.217 \\ 
  OR(X) & -0.077 & 0.961 & 0.964 & -0.018 & 0.873 & 0.873 \\ 
  OR(XV) & -0.077 & 0.961 & 0.964 & -0.008 & 0.873 & 0.873 \\ 
  RIPW(X) & 1.870 & 2.382 & 3.027 & -0.737 & 2.873 & 2.965 \\ 
  RAIPW(X,X) & -0.067 & 1.077 & 1.079 & -0.055 & 1.206 & 1.207 \\ 
  RAIPW(X,XV) & -0.062 & 1.080 & 1.081 & -0.058 & 1.200 & 1.201 \\ 
  ROR(X) & -0.072 & 0.975 & 0.978 & -0.021 & 0.891 & 0.891 \\ 
  ROR(XV) & -0.073 & 0.974 & 0.977 & -0.005 & 0.891 & 0.890 \\ 
   \hline
   IPW($X_{\_}$) & 1.630 & 1.909 & 2.51 & -1.289 & 2.115 & 2.476 \\ 
  AIPW($X_{\_}, XV$) & -0.073 & 0.989 & 0.992 & -0.057 & 1.067 & 1.068 \\ 
  AIPW($(X,X_{\_}V$) & -0.033 & 1.097 & 1.097 & -0.135 & 1.461 & 1.467 \\ 
  AIPW($X_{\_},X_{\_}V$) & 1.316 & 1.032 & 1.672 & -0.988 & 1.224 & 1.573 \\ 
  OR($X_{\_}$) & 1.300 & 0.985 & 1.631 & -0.898 & 0.939 & 1.299 \\ 
  RIPW($X_{\_}$) & 3.525 & 1.968 & 4.036 & -2.010 & 2.352 & 3.093 \\ 
  RAIPW($X_{\_}, XV$) & -0.067 & 1.049 & 1.050 & -0.051 & 1.110 & 1.111 \\ 
  RAIPW($X,X_{\_}V$) & 0.019 & 1.179 & 1.178 & -0.109 & 1.507 & 1.510 \\ 
  RAIPW($X_{\_},X_{\_}V$) & 1.325 & 1.104 & 1.724 & -1.056 & 1.276 & 1.656 \\ 
  ROR($X_{\_}V$) & 1.295 & 1.015 & 1.645 & -0.887 & 0.968 & 1.312 \\  \hline
\end{tabular}
\end{table}

\begin{table}
\caption{Bias, standard deviation and root mean squared error (times 10) of the estimates across simulations for the $\gamma$  strong-$\xi$ moderate under C-sym contamination.} 
\centering
\begin{tabular}{lrrrrrr} \hline
  ($\times 10$)  & bias($\hat{\mu}$) & sd($\hat{\mu})$ & $\sqrt{\mbox{mse}(\hat{\mu})}$ & bias($\hat{\sigma}$) & sd($\hat{\sigma})$ & $\sqrt{\mbox{mse}(\hat{\sigma})}$ \\ \hline
IPW(X) & -8.659 & 3.197 & 9.230 & 15.638 & 2.824 & 15.891 \\ 
  AIPW(X,X) & -8.839 & 1.576 & 8.978 & 15.866 & 2.355 & 16.039 \\ 
  AIPW(X,XV) & -8.831 & 1.551 & 8.966 & 15.864 & 2.331 & 16.035 \\ 
  OR(X) & -8.848 & 1.470 & 8.969 & 18.357 & 1.244 & 18.399 \\ 
  OR(XV) & -8.844 & 1.458 & 8.964 & 18.44 & 1.222 & 18.481 \\ 
  RIPW(X) & 3.505 & 3.877 & 5.226 & -3.540 & 5.056 & 6.170 \\ 
  RAIPW(X,X) & 0.326 & 2.077 & 2.101 & -2.011 & 3.159 & 3.744 \\ 
  RAIPW(X,XV) & 0.315 & 2.095 & 2.117 & -2.050 & 3.503 & 4.058 \\ 
  ROR(X) & 0.015 & 1.280 & 1.280 & 0.328 & 1.097 & 1.145 \\ 
  ROR(XV) & 0.015 & 1.251 & 1.250 & 0.231 & 1.015 & 1.040 \\ 
   \hline
   IPW($X_{\_}$) & -6.586 & 2.466 & 7.032 & 15.127 & 2.494 & 15.331 \\ 
  AIPW($X_{\_}, XV$) & -8.865 & 1.528 & 8.996 & 16.217 & 2.212 & 16.367 \\ 
  AIPW($(X,X_{\_}V$) & -8.806 & 1.617 & 8.953 & 15.846 & 2.317 & 16.014 \\ 
  AIPW($X_{\_},X_{\_}V$) & -7.570 & 1.53 & 7.723 & 15.636 & 2.235 & 15.795 \\ 
  OR($X_{\_}$) & -7.559 & 1.441 & 7.695 & 17.888 & 1.216 & 17.929 \\ 
  RIPW($X_{\_}$) & 5.760 & 2.793 & 6.401 & -5.254 & 3.817 & 6.493 \\ 
  RAIPW($X_{\_}, XV$) & 0.164 & 1.921 & 1.927 & -1.863 & 2.864 & 3.415 \\ 
  RAIPW($X,X_{\_}V$) & 0.417 & 2.096 & 2.136 & -2.052 & 3.172 & 3.776 \\ 
  RAIPW($X_{\_},X_{\_}V$) & 1.618 & 2.046 & 2.607 & -2.860 & 3.260 & 4.335 \\ 
  ROR($X_{\_}V$) & 1.359 & 1.262 & 1.854 & -0.493 & 1.059 & 1.168 \\ \hline
\end{tabular}
\end{table}

\begin{table}
\caption{Bias, standard deviation and root mean squared error (times 10) of the estimates across simulations for the $\gamma$  moderate-$\xi$ strong under C-sym contamination.} 
\centering
\begin{tabular}{lrrrrrr} \hline
  ($\times 10$)  & bias($\hat{\mu}$) & sd($\hat{\mu})$ & $\sqrt{\mbox{mse}(\hat{\mu})}$ & bias($\hat{\sigma}$) & sd($\hat{\sigma})$ & $\sqrt{\mbox{mse}(\hat{\sigma})}$ \\ \hline
IPW(X) & -9.203 & 1.891 & 9.395 & 14.010 & 1.716 & 14.114 \\ 
  AIPW(X,X) & -9.240 & 1.594 & 9.376 & 14.064 & 1.358 & 14.13 \\ 
  AIPW(X,XV) & -9.226 & 1.544 & 9.354 & 14.077 & 1.299 & 14.137 \\ 
  OR(X) & -9.247 & 1.619 & 9.388 & 15.988 & 1.184 & 16.032 \\ 
  OR(XV) & -9.233 & 1.567 & 9.365 & 16.074 & 1.134 & 16.114 \\ 
  RIPW(X) & 6.327 & 2.284 & 6.726 & -2.057 & 3.164 & 3.773 \\ 
  RAIPW(X,X) & -0.233 & 2.023 & 2.035 & 0.451 & 2.245 & 2.289 \\ 
  RAIPW(X,XV) & -0.359 & 1.899 & 1.932 & 0.574 & 2.079 & 2.156 \\ 
  ROR(X) & 0.001 & 1.691 & 1.690 & 0.442 & 1.331 & 1.402 \\ 
  ROR(XV) & 0.049 & 1.594 & 1.594 & 0.127 & 1.185 & 1.191 \\ 
   \hline
   IPW($X_{\_}$) & -7.678 & 1.819 & 7.890 & 13.834 & 1.626 & 13.929 \\ 
  AIPW($X_{\_}, XV$) & -9.233 & 1.550 & 9.362 & 14.301 & 1.308 & 14.361 \\ 
  AIPW($(X,X_{\_}V$) & -9.213 & 1.548 & 9.342 & 14.069 & 1.300 & 14.129 \\ 
  AIPW($X_{\_},X_{\_}V$) & -8.456 & 1.534 & 8.594 & 14.074 & 1.298 & 14.134 \\ 
  OR($X_{\_}$) & -8.457 & 1.551 & 8.598 & 15.841 & 1.129 & 15.881 \\ 
  RIPW($X_{\_}$) & 7.900 & 2.238 & 8.211 & -2.929 & 2.902 & 4.122 \\ 
  RAIPW($X_{\_}, XV$) & -0.428 & 1.898 & 1.945 & 0.610 & 2.054 & 2.141 \\ 
  RAIPW($X,X_{\_}V$) & -0.254 & 1.942 & 1.957 & 0.496 & 2.216 & 2.270 \\ 
  RAIPW($X_{\_},X_{\_}V$) & 0.559 & 1.907 & 1.986 & 0.001 & 2.132 & 2.130 \\ 
  ROR($X_{\_}V$) & 0.878 & 1.614 & 1.836 & -0.052 & 1.245 & 1.245 \\ \hline
\end{tabular}
\end{table}

\begin{table}
\caption{Bias, standard deviation and root mean squared error (times 10) of the estimates across simulations for the $\gamma$  strong-$\xi$ strong under C-sym contamination.} 
\centering
\begin{tabular}{lrrrrrr}\hline
  ($\times 10$)  & bias($\hat{\mu}$) & sd($\hat{\mu})$ & $\sqrt{\mbox{mse}(\hat{\mu})}$ & bias($\hat{\sigma}$) & sd($\hat{\sigma})$ & $\sqrt{\mbox{mse}(\hat{\sigma})}$ \\ \hline
IPW(X) & -9.009 & 3.966 & 9.842 & 13.714 & 3.312 & 14.108 \\ 
  AIPW(X,X) & -9.236 & 1.859 & 9.421 & 14.043 & 2.359 & 14.240 \\ 
  AIPW(X,XV) & -9.218 & 1.751 & 9.383 & 14.033 & 2.203 & 14.205 \\ 
  OR(X) & -9.230 & 1.765 & 9.397 & 16.753 & 1.474 & 16.817 \\ 
  OR(XV) & -9.226 & 1.704 & 9.382 & 16.829 & 1.382 & 16.886 \\ 
  RIPW(X) & 4.446 & 5.061 & 6.734 & -3.718 & 6.399 & 7.398 \\ 
  RAIPW(X,X) & 0.337 & 2.369 & 2.392 & -2.219 & 3.233 & 3.920 \\ 
  RAIPW(X,XV) & 0.161 & 2.101 & 2.106 & -1.978 & 2.717 & 3.360 \\ 
  ROR(X) & 0.007 & 1.663 & 1.663 & 0.494 & 1.447 & 1.528 \\ 
  ROR(XV) & 0.022 & 1.561 & 1.560 & 0.177 & 1.202 & 1.215 \\ 
   \hline
   IPW($X_{\_}$) & -6.404 & 3.051 & 7.094 & 12.913 & 2.78 & 13.208 \\ 
  AIPW($X_{\_}, XV$) & -9.250 & 1.742 & 9.413 & 14.407 & 2.127 & 14.563 \\ 
  AIPW($(X,X_{\_}V$) & -9.193 & 1.808 & 9.369 & 14.015 & 2.183 & 14.184 \\ 
  AIPW($X_{\_},X_{\_}V$) & -7.955 & 1.735 & 8.141 & 13.766 & 2.150 & 13.932 \\ 
  OR($X_{\_}$) & -7.940 & 1.678 & 8.115 & 16.224 & 1.373 & 16.282 \\ 
  RIPW($X_{\_}$) & 7.559 & 3.654 & 8.395 & -6.501 & 4.926 & 8.155 \\ 
  RAIPW($X_{\_}, XV$) & 0.050 & 2.018 & 2.018 & -1.921 & 2.374 & 3.053 \\ 
  RAIPW($X,X_{\_}V$) & 0.280 & 2.292 & 2.308 & -2.051 & 3.222 & 3.818 \\ 
  RAIPW($X_{\_},X_{\_}V$) & 1.637 & 1.942 & 2.539 & -3.196 & 2.310 & 3.943 \\ 
  ROR($X_{\_}V$) & 1.366 & 1.564 & 2.076 & -0.574 & 1.233 & 1.360 \\ \hline
\end{tabular}
\end{table}

\begin{table}
\caption{Bias, standard deviation and root mean squared error (times 10) of the estimates across simulations for the $\gamma$  moderate-$\xi$ no under C-sym contamination.} 
\centering
\begin{tabular}{lrrrrrr}\hline
 ($\times 10$)  & bias($\hat{\mu}$) & sd($\hat{\mu})$ & $\sqrt{\mbox{mse}(\hat{\mu})}$ & bias($\hat{\sigma}$) & sd($\hat{\sigma})$ & $\sqrt{\mbox{mse}(\hat{\sigma})}$ \\ \hline
IPW(X) & -8.468 & 1.178 & 8.549 & 18.104 & 1.437 & 18.161 \\ 
  AIPW(X,X) & -8.489 & 1.007 & 8.548 & 18.126 & 1.358 & 18.177 \\ 
  AIPW(X,XV) & -8.489 & 1.013 & 8.550 & 18.129 & 1.357 & 18.179 \\ 
  OR(X) & -8.499 & 1.029 & 8.561 & 19.621 & 0.938 & 19.643 \\ 
  OR(XV) & -8.499 & 1.034 & 8.562 & 19.713 & 0.935 & 19.736 \\ 
  RIPW(X) & 2.265 & 1.537 & 2.737 & -2.104 & 2.567 & 3.318 \\ 
  RAIPW(X,X) & 0.187 & 1.098 & 1.113 & -1.274 & 0.954 & 1.591 \\ 
  RAIPW(X,XV) & 0.192 & 1.105 & 1.121 & -1.276 & 0.956 & 1.594 \\ 
  ROR(X) & 0.003 & 0.987 & 0.986 & 0.287 & 0.792 & 0.842 \\ 
  ROR(XV) & 0.003 & 0.990 & 0.990 & 0.295 & 0.794 & 0.847 \\ 
   \hline
   IPW($X_{\_}$) & -7.559 & 1.125 & 7.642 & 18.094 & 1.391 & 18.148 \\ 
  AIPW($X_{\_}, XV$) & -8.498 & 1.019 & 8.559 & 18.308 & 1.334 & 18.357 \\ 
  AIPW($(X,X_{\_}V$) & -8.476 & 1.022 & 8.538 & 18.121 & 1.356 & 18.171 \\ 
  AIPW($X_{\_},X_{\_}V$) & -7.721 & 1.010 & 7.787 & 18.137 & 1.326 & 18.185 \\ 
  OR($X_{\_}$) & -7.725 & 1.024 & 7.792 & 19.527 & 0.931 & 19.550 \\ 
  RIPW($X_{\_}$) & 3.180 & 1.408 & 3.478 & -2.412 & 2.322 & 3.348 \\ 
  RAIPW($X_{\_}, XV$) & 0.157 & 1.100 & 1.111 & -1.271 & 0.936 & 1.578 \\ 
  RAIPW($X,X_{\_}V$) & 0.268 & 1.126 & 1.157 & -1.319 & 1.036 & 1.677 \\ 
  RAIPW($X_{\_},X_{\_}V$) & 1.052 & 1.108 & 1.527 & -1.607 & 0.995 & 1.890 \\ 
  ROR($X_{\_}V$) & 0.828 & 1.009 & 1.305 & 0.192 & 0.857 & 0.877 \\ \hline
\end{tabular}
\end{table}

\begin{table}
\caption{Bias, standard deviation and root mean squared error (times 10) of the estimates across simulations for the $\gamma$ strong-$\xi$ no under C-sym contamination.} 
\centering
\begin{tabular}{lrrrrrr}\hline
 ($\times 10$)  & bias($\hat{\mu}$) & sd($\hat{\mu})$ & $\sqrt{\mbox{mse}(\hat{\mu})}$ & bias($\hat{\sigma}$) & sd($\hat{\sigma})$ & $\sqrt{\mbox{mse}(\hat{\sigma})}$ \\ \hline
IPW(X) & -8.295 & 2.476 & 8.656 & 17.833 & 2.532 & 18.012 \\ 
  AIPW(X,X) & -8.427 & 1.396 & 8.542 & 17.967 & 2.432 & 18.131 \\ 
  AIPW(X,XV) & -8.430 & 1.406 & 8.546 & 17.977 & 2.432 & 18.141 \\ 
  OR(X) & -8.464 & 1.255 & 8.556 & 20.175 & 1.078 & 20.204 \\ 
  OR(XV) & -8.461 & 1.265 & 8.555 & 20.261 & 1.082 & 20.290 \\ 
  RIPW(X) & 2.102 & 3.607 & 4.173 & -2.039 & 6.092 & 6.421 \\ 
  RAIPW(X,X) & -0.162 & 3.504 & 3.506 & -0.974 & 6.999 & 7.063 \\ 
  RAIPW(X,XV) & -0.137 & 3.527 & 3.528 & -1.079 & 7.302 & 7.378 \\ 
  ROR(X) & 0.010 & 0.977 & 0.977 & 0.311 & 0.839 & 0.894 \\ 
  ROR(XV) & 0.011 & 0.977 & 0.976 & 0.318 & 0.840 & 0.898 \\ 
   \hline
   IPW($X_{\_}$) & -6.758 & 1.948 & 7.033 & 17.606 & 2.367 & 17.764 \\ 
  AIPW($X_{\_}, XV$) & -8.469 & 1.370 & 8.580 & 18.303 & 2.298 & 18.446 \\ 
  AIPW($(X,X_{\_}V$) & -8.405 & 1.477 & 8.533 & 17.958 & 2.415 & 18.119 \\ 
  AIPW($X_{\_},X_{\_}V$) & -7.177 & 1.384 & 7.309 & 17.819 & 2.316 & 17.968 \\ 
  OR($X_{\_}$) & -7.177 & 1.262 & 7.287 & 19.791 & 1.079 & 19.821 \\ 
  RIPW($X_{\_}$) & 3.716 & 3.021 & 4.788 & -3.129 & 5.575 & 6.391 \\ 
  RAIPW($X_{\_}, XV$) & 0.108 & 2.945 & 2.946 & -1.724 & 6.497 & 6.719 \\ 
  RAIPW($X,X_{\_}V$) & -0.031 & 3.472 & 3.471 & -1.088 & 7.157 & 7.236 \\ 
  RAIPW($X_{\_},X_{\_}V$) & 1.321 & 2.923 & 3.206 & -1.913 & 5.474 & 5.796 \\ 
  ROR($X_{\_}V$) & 1.354 & 0.999 & 1.683 & -0.332 & 0.893 & 0.953 \\ \hline
\end{tabular}
\end{table}

\clearpage
\section{Application: supplementary results}
\label{applicationsupp.section}

Table \ref{auxmodel.tab} describes the fitted auxiliary models explaining dropout ($R_i$) and BMI change (outcome $Z_{2i}$) with covariates: (from the health
examinations) measured BMI (bbmi, kg/m$^2$), self reported health (srh, 1 if positive self reported health, zero otherwise) and tobacco use (tob, 1 if cigaretes and/or snus user, 0 otherwise); (from Statistics
Sweden registers) education level (educ, 1 if more than 9 years education, 0 otherwise), number of children under 3 years of age (nrchild3), log annual 
earnings (logearn), annual parental benefits (parbenef), annual sick leave benefits (sickbenef), annual unemployment benefits (unempbenef),
urban living (urban, 1 if urban living area, 0 otherwise); (from the hospitalisation register)
hospitalisation days (no hospitalisation is reference, hosp13 for 1 to 3 days hospitalisation during baseline year, hosp4M for more than 3 days hospitalisation). 

Both OLS/Maximum likelihood (using logit link for the binary indicator $R_i$) and robust estimators are used; see Sections \ref{roblogistic.sec} and \ref{robreg.section}.   
The logistic regression fit and its robust GLM version give similar results. On the other hand, there are clear difference between OLS and its robust version fits. In particular, "self reported health" is significant at 1 \% level in the OLS fit explaining BMI change while it is not anymore (at the 10\% level) with the robust fit. Conversely "sick benefits" becomes clearly significant (1\% level) with the robust estimation. BMI at baseline while significant in both cases, has 10 times lower explaining effect in the robust fit.

The implementation of RIPW, RAIPW and ROR are based on these robust fits of auxiliary models. We use $c_\mu=4$ and $c_\sigma=5$ for the results of Table \ref{results.tab}. These corresponds to values tuned in the simulation study, see Table \ref{tuningconstants.table}, and varying $c_\mu$ and $c_\sigma$ within [2,6] and [3,7] respectively for $c_\mu$ and $c_\sigma$ did not change the results of Table \ref{results.tab} notably.


  \begin{table}
\caption{  \label{auxmodel.tab}  Estimated auxiliary models explaining dropout and change in BMI using the covariates listed in Section  \ref{bmi.sec} of main paper (s.e. in parantheses).} 
  \centering 
\begin{tabular}{@{\extracolsep{5pt}}lcccc}  
\\[-1.8ex]
\hline 
 & \multicolumn{4}{c}{{Dependent variable}} \\ 
\cline{2-5} 
\\[-2.8ex]  & \multicolumn{2}{c}{$R_i=0$: dropout} & \multicolumn{2}{c}{$Z_{2i}$: BMI change} \\ 
\\[-2.8ex] & \multicolumn{1}{c}{{logistic}} & \multicolumn{1}{c}{{robust}} & \multicolumn{1}{c}{{OLS}} & \multicolumn{1}{c}{{robust}} \\ 
\hline \\[-2.8ex] 
  bbmi & -0.033$^{***}$ & -0.033$^{***}$ & -0.203$^{***}$ & -0.029$^{***}$ \\ 
  & (0.007) & (0.007) & (0.010) & (0.007) \\ 
  srh & 0.146$^{**}$ & 0.153$^{**}$ & -0.271$^{***}$ & -0.045 \\ 
  & (0.064) & (0.064) & (0.093) & (0.058) \\ 
  tob & -0.119$^{**}$ & -0.120$^{**}$ & 0.310$^{***}$ & 0.233$^{***}$ \\ 
  & (0.053) & (0.054) & (0.075) & (0.047) \\ 
  educ & 0.026 & 0.028 & -0.188$^{*}$ & -0.040 \\ 
  & (0.069) & (0.070) & (0.101) & (0.063) \\ 
  nrchild3 & -0.156$^{**}$ & -0.168$^{**}$ & -0.046 & 0.047 \\ 
  & (0.072) & (0.073) & (0.105) & (0.065) \\ 
  logearn & 0.050$^{***}$ & 0.049$^{***}$ & -0.027 & 0.007 \\ 
  & (0.013) & (0.013) & (0.020) & (0.013) \\ 
  parbenef & 0.226$^{**}$ & 0.237$^{**}$ & -0.192 & -0.104 \\ 
  & (0.096) & (0.098) & (0.129) & (0.081) \\ 
  sickbenef & 0.366$^{***}$ & 0.356$^{***}$ & 0.167$^{*}$ & 0.258$^{***}$ \\ 
  & (0.068) & (0.069) & (0.091) & (0.057) \\ 
  unempbenef & -0.300$^{***}$ & -0.298$^{***}$ & 0.017 & -0.062 \\ 
  & (0.073) & (0.073) & (0.115) & (0.072) \\ 
  urban & -0.042 & -0.038 & 0.050 & 0.037 \\ 
  & (0.053) & (0.054) & (0.076) & (0.047) \\ 
  hosp13 & -0.433$^{***}$ & -0.436$^{***}$ & 0.179 & 0.036 \\ 
  & (0.138) & (0.139) & (0.220) & (0.137) \\ 
  hosp4M & -0.427$^{**}$ & -0.400$^{**}$ & 0.569$^{**}$ & -0.137 \\ 
  & (0.181) & (0.182) & (0.282) & (0.176) \\ 
  Constant & 1.466$^{***}$ & 1.454$^{***}$ & 6.860$^{***}$ & 1.923$^{***}$ \\ 
  & (0.227) & (0.229) & (0.334) & (0.209) \\ 
 \hline \\[-2.8ex] 
Observations & \multicolumn{1}{c}{7,555} & \multicolumn{1}{c}{7,555} & \multicolumn{1}{c}{5,553} & \multicolumn{1}{c}{5,553} \\ 
\hline  \\[-2.8ex] 
\textit{Note:}  & \multicolumn{4}{r}{$^{*}$p$<$0.1; $^{**}$p$<$0.05; $^{***}$p$<$0.01} \\ 
\end{tabular} 
\end{table}

\end{document}